\documentclass[a4paper,11pt]{article}
\usepackage{authblk}
\usepackage{breakcites}
\usepackage{newtxtext}
\usepackage[margin=1in]{geometry}

\usepackage{graphicx}
\usepackage{caption}
\usepackage{amsmath,amsfonts,amssymb,amsthm}
\usepackage{bbm}

\usepackage{algorithm}
\usepackage{booktabs}

\newtheorem{theorem}{Theorem}
\newtheorem{corollary}{Corollary}
\newtheorem{appr}{Approximation}

\newcommand{\dd}{\mathrm{d}}
\newcommand{\Var}{\operatorname{Var}}
\newcommand{\E}{\mathbb{E}}
\newcommand{\N}{\mathcal{N}}

\begin{document}

\title{A posteriori stochastic correction of reduced models in delayed acceptance MCMC, with application to multiphase subsurface inverse problems}

\author[1]{Tiangang Cui\thanks{tiangang.cui@monash.edu}}
\author[2]{Colin Fox\thanks{colin.fox@otago.ac.nz}} 
\author[3]{Michael~J.~O'Sullivan\thanks{m.osullivan@auckland.ac.nz}}

\affil[1]{School of Mathematical Sciences, Monash University, Melbourne, Australia}
\affil[2]{Department of Physics, University of Otago, Dunedin, New Zealand}
\affil[3]{Department of Engineering Science, The University of Auckland, Auckland, New Zealand} 

\maketitle

\begin{abstract}
Sample-based Bayesian inference provides a route to uncertainty quantification in the geosciences, and inverse problems in general, though is very computationally demanding in the na\"ive form that requires simulating an accurate computer model at each iteration. 
We present a new approach that constructs a stochastic correction to the error induced by a reduced model, with the correction improving as the algorithm proceeds. This enables sampling from the correct target distribution at reduced computational cost per iteration, as in existing delayed-acceptance schemes, while avoiding appreciable loss of statistical efficiency that necessarily occurs when using a reduced model. 
Use  of the stochastic correction significantly reduces the computational cost of estimating quantities of interest within desired uncertainty bounds. In contrast, existing schemes that use a reduced model directly as a surrogate do not actually improve computational efficiency in our target applications.
We build on recent simplified conditions for adaptive Markov chain Monte Carlo algorithms to give practical approximation schemes and algorithms with guaranteed convergence.
The efficacy of this new approach is demonstrated in two computational examples, including calibration of a large-scale numerical model of a real geothermal reservoir, that show good computational and statistical efficiencies on both synthetic and measured data sets.
\end{abstract}

\section{Introduction\label{sec:intro}}

Characterizing subsurface flow in aquifers, geothermal and petroleum reservoirs often requires inferring unobserved subsurface properties from indirect observations made on the underlying system. This is a typical inverse problem. 
There are several fundamental difficulties associated with this process: data are corrupted by noise in the measurement process and often sparsely measured, the forward model only has a limited range of accuracy in representing the underlying system, and the parameters of interest are usually spatially distributed and highly heterogeneous, e.g., permeabilities and conductivities.
These features make the inverse problem ill-posed~\cite{hadamard, stats_inverse}.
This implies that there exists a range of feasible parameters that are consistent with the data, and hence a range of possible model predictions. 
Ill-posedness also causes the correlations between fitted parameters to be extremely high, and the model output to occupy some low dimensional manifold in data space. 
These properties make the traditional, deterministic solution to the inverse problem very sensitive to measurement error and model error.

In this paper, we develop the `solution' to the inverse problem in the Bayesian inference framework, and present algorithmic advances for computing inferential solutions. Solutions to inverse problems under the Bayesian framework is well established, 
with comprehensive examples including characterization of subsurface properties of aquifers and petroleum reservoirs such as \cite{hydromcmc_oliver1997,  model_higdon2002, model_higdon2003,  hydromcmc_mondal_2010}, remote sensing such as \cite{cornford_inversion2004,haario_inversion2004} and impedance imaging such as \cite{mcmc_resistivity,model_watzenig2009}.
By incorporating prior information and modelling uncertainties in the data observation process, all information regarding the solution to the inverse problem is coded into the posterior distribution over unknown parameters conditioned on the measured data. 

Sample-based inference proceeds by computing Monte Carlo estimates of posterior statistics, to give `solutions' and quantified uncertainties, using samples drawn from the posterior distribution via Markov chain Monte Carlo (MCMC) sampling. However, the applicability of this approach is fundamentally limited by the computational cost of the forward map and errors associated with the data modelling process.
Difficulties arise from the high dimensionality of the model parameters, and computational demands of the forward map, which is usually dictated by numerical solutions of nonlinear partial differential equations (PDEs).
In practice, it is computationally prohibitive to run standard MCMC for many iterations in realistic geophysical inverse problems.

Computational efficiency can be significantly improved by exploiting reduced models and recent advances in model errors.
We present a new adaptive delayed acceptance (ADA) algorithm  that
%
is based on the delayed acceptance (DA) MCMC algorithm \cite{ChristenFox2005} that takes advantage of a reduced model and the associated approximate posterior to speed up the computation. 
Coupled with recent advances in stochastic models for model errors \cite{model_ohagan, model_kaipio1} and state-of-the-art techniques in adaptive MCMC algorithms \cite{adaptive_roberts2007}, ADA dynamically improves the accuracy of the approximate posterior distribution.

The adaptation used here is substantially different from the classical adaptive mesh refinement of finite difference or finite element methods. 
Rather, we start with a given reduced model for the forward map, in addition to an expensive, exact forward map that defines the desired target distribution and solutions. We implement a new usage of adaptive MCMC methods by learning a stochastic correction to the reduced model, thereby bringing the corrected approximate posterior closer to the exact posterior, as well as the usual adaptation of the proposal distribution. 

The use of reduced, deterministic models is well developed in large-scale inverse problems, particularly in optimization for evaluating regularized inverses. Our interest is in using reduced models to improve efficiency of MCMC for fitting Bayesian hierarchical formulations of inverse problems. A reduced model induces an approximated likelihood function, and hence approximate posterior distribution. The approximate posterior distribution may be used to modify a proposal distribution by first performing an accept-reject step using the approximate posterior distribution. Subsequently using the modified proposal in a standard Metropolis--Hastings accept-reject step with the exact posterior distribution  guarantees convergence to the desired, exact target distribution. When the reduced model has reduced computational cost, compared to the exact computational model, computational expense is potentially reduced since only those proposals that are accepted by the approximation are passed on for the exact, expensive model computation. 

The simplest case of using a fixed reduced model, that induces a fixed `surrogate' target distribution, was introduced by Liu~\cite{mcmc_liu} as the `surrogate transition method', and later reintroduced as preconditioned or two-stage MCMC~\cite{Efendiev}. This method uses Approximation~\ref{appr:coarse} in Sec.~\ref{sec:approx}. Christen and Fox~\cite{ChristenFox2005} introduced the more general `delayed acceptance' algorithm that allows the approximation to depend on the state of the MCMC, and demonstrated a speed-up in an inverse problem using a local linearization as the state-dependent reduced model. The delayed acceptance algorithm is presented in Section~\ref{sec:da}. 

The use of an approximate model in the delayed acceptance MCMC necessarily reduces statistical efficiency~\cite{ChristenFox2005}, compared to the MCMC that uses the exact target only. That is, more iterations are required to calculate quantities of interest to within a desired accuracy, asymptotic in the length of the chain. Without care, this increase in required number of iterations can offset the decrease in computation per iteration to the point where computational cost is not actually reduced, despite the extra programming effort. 
In the computed example in Section~\ref{sec:1d} we find that the use of a state-dependent reduced model is critical in improving computational efficiency in our target application; the same conclusion was reached in the setting of a large-data inverse problem in economics~\cite{Quiroz2018}. 

The main contribution in this work is extending the delayed acceptance MCMC to allow correcting the deterministic reduced model 
by learning a stochastic correction to the \textit{a posteriori} error induced by the reduced model, i.e.,  the difference between exact and approximated likelihood functions over states distributed as the posterior distribution. Stochastic corrections to the likelihood defined by a reduced model have been used previously, though the construction has been made \textit{a priori}, that is, before exploration of the posterior distribution is performed. Here, we show that the construction may be performed \textit{a posteriori}, as the posterior distribution is being explored, resulting in significant computational gains. In Section~\ref{sec:aem} we follow the development of the approximation error model (AEM) introduced by Kaipio \& Somersalo for Bayesian-inspired regularization~\cite{stats_inverse}, though the same construction may be found in the field of computer model calibration~\cite{model_ohagan}. We also draw on the techniques of adaptive MCMC~\cite{adaptive_haario2001,adaptive_roberts2007}, that allows the transition kernel of the Markov chain to be modified as iterations progress, in a way that guarantees ergodicity for the desired posterior distribution.
The AEM is presented in Section~\ref{sec:aem}, with improvements developed subsequently in Section~\ref{sec:approx}. Standard adaptive MCMC is presented in Section~\ref{sec:adapt}, with our novel adaptation of the stochastic correction presented in Section~\ref{sec:adamh}.

In Section~\ref{sec:approx} we present details of five approximation schemes: a fixed reduced model giving the surrogate transition method (Approx.~\ref{appr:coarse}); a fixed reduced model with AEM built over the prior distribution as in~\cite{stats_inverse}  (Approx.~\ref{appr:aempr}); a fixed reduced model with AEM built adaptively over the posterior distribution (Approx.~\ref{appr:aempost}); a zeroth-order correction of the fixed reduced model that gives a state-dependent approximation (Approx.~\ref{appr:local}); and the state-dependent approximation with stochastic correction built adaptively over the posterior distribution (Approx.~\ref{appr:aemsd}). Approximations ~\ref{appr:aempost}, \ref{appr:local}, and \ref{appr:aemsd} are new methods. Each of these five approximations have essentially the same computational cost per iteration, as the cost of the corrections are negligible in our target applications. Yet the statistical efficiency of the resulting MCMC is very different. Compared to the unmodified, exact MCMC, the surrogate transition method (Approx.~\ref{appr:coarse}) requires many times more iterations to estimate parameters to within a desired accuracy, while Approximation~\ref{appr:aemsd} requires virtually the same number of iterations. Statistical efficiency and computational efficiency are formally defined in Sections~\ref{sec:cme} and \ref{sec:da}, respectively. This quantification of the trade-off between statistical efficiency and computing time of MCMC is applicable to all sampled-based solutions of inverse problems. A comparative study of computational cost and efficiency is presented in a simulated analysis of a well discharge test in Section~\ref{sec:1d} that demonstrates the progression in computational efficiency from  Approximation~\ref{appr:coarse} to Approximation~\ref{appr:aemsd}.

The new adaptive delayed acceptance (ADA) algorithm, that builds the \textit{a posteriori} stochastic correction and can implement all approximations, is presented in Section~\ref{sec:adamh}. A proof that ADA is ergodic for the target distribution is given in Section~\ref{sec:ergodic}, to guarantee that Monte Carlo estimates evaluated over the chain converge to the true, posterior value. Proofs of ergodicity are crucial for correctness of an MCMC since convergence is in distribution so cannot be diagnosed from a single state, unlike optimization algorithms where a local check of zero gradient is possible. Ergodicity of adaptive MCMC algorithms, such as ADA, is even more delicate since it is easy to specify an adaptive MCMC that is incorrect, that is, Monte Carlo estimates fail to converge to the correct value. To emphasize this point~\cite{adaptive_roberts2007} presents examples of seemingly straightforward adaptive MCMC algorithms that are actually incorrect. Our proof of ergodicity utilizes the simplified requirements of simultaneous uniform ergodicity and diminishing adaptation introduced in~\cite{adaptive_roberts2007}; see Theorem~\ref{theo1}. The general framework allowed by the proof of ergodicity in Theorems~\ref{theo2} and \ref{theo3} is not limited by the current form of reduced models and posterior approximations used in this paper. The general regularity conditions of our ergodicity theorem open the door to building other goal-oriented reduced models for high-dimensional inverse problems.

A further innovation presented in Section~\ref{sec:gcam} is the grouped-component adaptive Metropolis (GCAM) proposal distribution 
that is suitable for high-dimensional inverse problems with black-box forward models. 


We validate ADA on two models of geothermal reservoirs. The first is a 1D homogeneous model with 7 unknown parameters, using synthetic transient data. This moderate-sized inverse problem allows extensive calculation of statistics providing a comparative study of efficiencies of algorithms and approximations. In the second example we predict the hot plume of a 3D multi-phase geothermal reservoir model by estimating the heterogeneous and anisotropic permeability distribution and the heterogeneous boundary conditions \cite{wrr_cfo_2011}.
This model has $10,049$ unknown parameters, and each model simulation requires about 30 to 50 minutes of CPU time. 
In both studies, ADA shows superior computational performance than the counterparts that do not use adaptation.

This paper is structured as follows: 
Section \ref{sec:bfpe} briefly sets out the Bayesian formulation of inverse problems, reviews existing sample-based inference including delayed-acceptance and adaptive MCMC algorithms, discusses computational and statistical efficiency, and presents the GCAM proposal used in this paper. 
Section \ref{sec:approx} presents deterministic and stochastic approximations to the forward map and posterior distribution that are used to reduce computational cost per iteration. Section~\ref{sec:postAEM} \textit{et seq.} constitute the new contributions in this paper.
Section \ref{sec:adamh} presents the ADA algorithm that utilizes the approximations developed in Section~\ref{sec:approx}, and gives a proof of convergence. 
Section \ref{sec:geo} presents two case studies of ADA on calibrating geothermal reservoir models. 
Section \ref{sec:dis} offers some conclusions and discussion.

\section{Bayesian formulation and posterior exploration}
\label{sec:bfpe}
In this section, we review existing sample-based Bayesian methods for inverse problems that are relevant to the ADA algorithm developed in Section~\ref{sec:adamh}. 
We also discuss the computational efficiency of MCMC, that is the total computing cost required to evaluate a quantity of interest to within a given error tolerance,  and present a new adaptive MCMC proposal used in our numerical examples.

\subsection{Bayesian formulation}
\label{sec:inv}

Suppose a physical system is simulated by a forward model $F(\cdot)$ and the unknowns of the physical system are parametrized by model parameters $\bf x$. 
The measurable features of the system are the image of the forward map at the true but unknown parameters $\bf x$, ${\bf d} = F({\bf x})$, called the true, noise-free data. Practical measurements $\tilde{\bf d}$ are a noisy version of  ${\bf d}$ subject to measurement errors and other sources of imprecision.
In an inverse problem, we wish to recover the unknown $\bf x$ from measurements $\tilde{\bf d}$, and to make predictions about properties of the physical system, such as future, unobserved data. 

Uncertainty in measured data $\tilde{\bf d}$ and in the model $F(\cdot)$ leads to uncertain estimates of parameters $\bf x$, and in subsequent predictions. What then is the range of permissible values, that is, the distribution over resulting estimates? The Bayesian formulation assigns probability distributions to each of the sources of error, and tracks the resulting distributions over quantities of interest. In our experience, this is the essential feature of the Bayesian method. 

The Bayesian framework quantifies the distribution over parameters, consistent with the measured data, by the posterior distribution with the unnormalized probability density function
\begin{equation} 
\pi_{\rm post}({\bf x}, \gamma, \delta \, |\, \tilde{\bf d}) \propto L(\tilde{\bf d} \, | \, {\bf x}, \gamma) \, \pi_{\rm prior}({\bf x} \,|\, \delta) \, \pi_{\rm prior}(\gamma) \, \pi_{\rm prior}(\delta),
\label{eq:bayes}
\end{equation} 
where $L(\tilde{\bf d} \, | \, {\bf x}, \gamma)$ is the likelihood function, giving the probability density of measuring data $\tilde{\bf d}$ in identical measurement setups defined by $ {\bf x}$ and $\pi_{\rm prior}({\bf x} \,|\, \delta)$ is the prior probability density function for model parameters ${\bf x}$.
In Eqn. \ref{eq:bayes}, we also introduce hyperparameters $\gamma$ and $\delta$ to represent modelling assumptions in the likelihood function and the stochastic model of ${\bf x}$, respectively. 
In a typical hierarchical Bayesian setting, these hyperparameters follow independent prior distributions $\pi_{\rm prior}(\gamma)$ and $\pi_{\rm prior}(\delta)$.

%
%

The likelihood function is determined by modelling the stochastic process that results in measured data for a given state of the system. Measured data is commonly assumed to be related to noise-free data by the additive error model
\begin{equation} 
\tilde{\bf d} = F({\bf x}) + \bf e,
\label{eq:cali}
\end{equation} 
where the random vector $\bf e$ captures the measurement noise and other uncertainties such as model error. 
%
%
When $\bf e$ follows a zero mean multivariate Gaussian distribution the resulting likelihood function has the form 
\begin{equation} 
L(\tilde{\bf d} \, | \, {\bf x}, \Sigma_{\bf e}) \propto \exp \left[ -\frac{1}{2}  \big( F({\bf x})-\tilde{\bf d} \big) ^\top \boldsymbol \Sigma_{\bf e}^{-1} \big( F({\bf x})-\tilde{\bf d}\big) \right],
\label{eq:llkd1}
\end{equation} 
where the hyperparameter $\gamma = \boldsymbol\Sigma_{\bf e}$ is the covariance matrix of the noise vector $\bf e$, 
with uncertainty in the covariance being modelled by the (hyper)prior distribution $\pi_{\rm prior}(\gamma)$. 
Note that evaluating the likelihood function for a particular ${\bf x}$ requires simulating the forward model $F(\bf x)$, which is a computationally taxing numerical simulation of a mathematical model of the physical system. 

The contribution of model error to the noise vector $\bf e$ is usually non-negligible, in practice. This may be caused by discretization error in the computer implementation of the mathematical forward model and/or wrong assumptions in the mathematical model. 
We follow \cite{model_higdon2003} who observe that it may not possible to separate the measurement noise and model error in the case that only single set of data is available. 
Thus, it is necessary to incorporate the modeller's judgments about the appropriate size and nature of the noise term $\bf e$. 
We also adopt the assumption of \cite{model_higdon2003} that the noise $\bf e$ follows a zero mean Gaussian distribution, giving the likelihood function in Eqn.~\ref{eq:llkd1}.

The primary unknowns ${\bf x}$ of a physical system can often be represented by spatially distributed parameters, for example, the coefficients in a partial differential equation modelling energy and/or mass propagation. Correspondingly,  representations and prior models are usefully drawn from spatial statistics, often in the form of hierarchical models \cite{hierarchical_bayes}; see \cite{rue_2003} for a review.   
%
Since the particulars of prior modelling are problem specific, we delay presenting detailed functional forms to the examples in Section~\ref{sec:geo}.

Because the interaction between the primary parameter $\bf x$ and hyperparameters $\Sigma_{\bf e}$ and $\delta$ introduces significant computational difficulties \cite{gmrf_rue}, here we set the value of hyperparameters based on expert opinion and field measurements. This yields the posterior density function
\begin{equation} 
\pi_{\rm post}({\bf x} \, | \, \tilde{\bf d}) \propto \exp\left[-\frac{1}{2}  \big( F({\bf x})-{\bf \tilde{d}} \big)^\top \boldsymbol\Sigma_{\bf e}^{-1} \big( F({\bf x})-{\bf \tilde{d}} \big)\right] \pi_{\rm prior}({\bf x}),
\label{eq:post}
\end{equation}
that is used throughout the remainder of this paper.

We summarize solutions to the inverse problem and the associated uncertainties as the expectation of some quantities of interest over the posterior. For instance, the mean of parameters, the variance of model states, and the credible intervals of model predictions.
Given a quantity of interest $g({\bf x})$, we calculate the Monte Carlo estimate, denoted by $\overline{g}_n$, of the posterior expectation as 
\begin{equation} 
\E[g] = \int g({\bf x}) \pi_{\rm post}({\bf x | {\bf \tilde{d}}}) \,\dd{\bf x}
                     \approx \overline{g}_n = \frac{1}{n}\sum_{i=1}^{n} g({\bf x}_{i}),
\label{eq:expect}
\end{equation} 
using $n$ samples drawn from the posterior distribution, i.e., ${\bf x}_{i}\sim \pi_{\rm post}({\bf x} \, | \, \tilde{\bf d})$. 
Thus, the task of estimating parameters or predictive values is reduced to the task of drawing samples from the posterior distribution; this defines sample-based Bayesian inference. We present algorithms for sampling from $\pi_{\rm post}$ in Sections~\ref{sec:mh} and \ref{sec:da}, and novel efficient methods in Section~\ref{sec:adamh}.

\subsection{Metropolis-Hastings Dynamics}
\label{sec:mh}

The basis of the sampling methods we develop is the  Metropolis Hasting (MH) algorithm  \cite{metropolis, mcmc_hastings}. One initializes the Markov chain at some starting state ${\bf x}_0$ 
and then iterate as in Alg.~\ref{alg:mh}. 

\begin{algorithm}[ht]
\caption{Metropolis-Hastings}
At iteration $n$, given ${\bf x}_n = {\bf x}$, then ${\bf x}_{n+1}$ is determined in the following way:
\begin{enumerate}
\item Propose new state ${\bf y}$ from some distribution $q\left({\bf x},\cdot\right)$.
\item With probability 
\begin{equation} 
  \alpha\left({\bf x, y}\right) = \min\left\{1, \frac{\pi_{\rm post}\left({\bf y} \, | \, \tilde{\bf d}\right) q\left({\bf y} ,{\bf x}\right)}
                     {\pi_{\rm post}\left({\bf x} \, | \, \tilde{\bf d}\right)q\left({\bf x} ,{\bf y}\right)}   
      \right\},
\end{equation} 
set ${\bf x}_{n+1} = {\bf y}$, otherwise ${\bf x}_{n+1} = {\bf x}$.
\end{enumerate}
\label{alg:mh}
\end{algorithm} 
%
%
The only choice one has in Alg.~\ref{alg:mh} is the choice of proposal distribution $q({\bf x, y})$. Together with the acceptance probability $\alpha\left({\bf x, y}\right)$, Alg.~\ref{alg:mh} defines a transition kernel
\begin{equation} 
K\left({\bf x, y}\right) = q\left({\bf x, y}\right)\,\alpha\left({\bf x, y}\right) + \left[ 1 - \int q\left({\bf x, z}\right)\,\alpha\left({\bf x, z}\right) d {\bf z} \right] \delta_{\bf x}({\bf y}),
\label{eq:transkernel}
\end{equation} 
that produces a Markov chain of correlated samples. In Eqn. \ref{eq:transkernel}, $\delta_{\bf x}({\bf y})$ is the Dirac delta function and the second term represents the probability of remaining at current state $\bf x$.
The choice of the proposal is largely arbitrary as long as it satisfies reversibility, recurrence, and irreducibility \cite{mcmc_robert}, and thus the resulting MH algorithm is {\it ergodic} and has $\pi_{\rm post}( \, \cdot\,| \, \tilde{\bf d})$ as the {\it unique stationary distribution}. 
An ergodic MH algorithm produces samples that converge in distribution to the posterior as the number of iterations $n\rightarrow\infty$.
%
%
After a burn-in period, in which the Markov chain effectively loses dependency on the starting state, samples from the chain may be substituted directly into the Monte Carlo estimate in Eqn.~\ref{eq:expect} to produce the estimate $\overline{g}_n$ of quantity $g$.

The choice of the proposal has a significant influence on the rate of  convergence of $\overline{g}_n$ to $g$, that depends on the degree of correlation \cite{Sokal,mcmc_geyer1992}. Markov chains that are fast to converge have lower correlation between adjacent samples.
Traditionally, the proposal distribution is chosen from some simple family of distributions, e.g., multivariate Gaussian, then manually tuned in a ``trial and error'' manner to optimize the rate of convergence. We present automatic, adaptive methods for tuning the proposal in Sections~\ref{sec:adapt} and \ref{sec:adamh}.

%
%

\subsection{Convergence and Statistical Efficiency}
\label{sec:cme}

Sample-based inference returns $\overline{g}_n$ in Eqn.~\ref{eq:expect}, which is an estimate of the quantity of interest with some Monte Carlo error due to $n$ being finite. 
%
Under mild conditions, convergence of the Monte Carlo estimate $\overline{g}_n$, in Eqn.~\ref{eq:expect}, is guaranteed by a central limit theorem (CLT) \cite{clt_kipnis1986} that gives
\(\overline{g}_n-\E\left[ g \right]\sim\N\left(0,\Var(\overline{g}_n)\right) \) as $n\rightarrow\infty$. Here, $\N\left(\mu,\sigma^2\right)$ denotes the normal (or Gaussian) distribution with mean $\mu$ and standard deviation $\sigma$. Hence, asymptotic in sample size $n$, $\overline{g}_n\rightarrow\E\left[ g \right]$ almost surely, with accuracy $\sqrt{\Var(\overline{g}_n)}$ when $n$ samples are used. 
When the samples ${\bf x}_{i}$ are independent, it follows that 
\[ 
\Var(\overline{g}_n) = \frac{\Var(g)}{n}. 
\] 
Since $\Var(g)$ depends only on the quantity being estimated, this sets the fewest number of iterations that a practical MCMC algorithm will require to achieve a given accuracy.

When the ${\bf x}_{i}$ are samples from a correlated Markov chain, as generated by any of the sampling algorithms discussed in this paper, instead (for  large $n$)
\[ 
\Var(\overline{g}_n) = \frac{\Var(g)}{n}\left(1+2\sum_{i=1}^\infty \rho_{gg}(i) \right), \] 
where $\rho_{gg}(i)$ is the autocorrelation coefficient for the chain in $g$ at lag $i$. Thus, the rate of variance reduction, compared to independent samples, is reduced by the factor
\[ \tau = \left(1+2\sum_{i=1}^\infty \rho_{gg}(i) \right), \] 
which is the integrated autocorrelation time (IACT) for the statistic $g$~\cite{Sokal}. 
We can think of $\tau\geq1$ being the length of the correlated chain that produces the same variance reduction as one independent sample. We call the quantity $1/\tau$ the {\it statistical efficiency}, while $n / \tau$ gives the number of effective (independent) samples for a Markov chain of length $n$.

\subsection{Delayed Acceptance and Computational Efficiency}
\label{sec:da}
We define the {\it computational efficiency} of MH to be the effective sample size per unit CPU time. Hence, as shown in the previous section, this is proportional to the rate of variance reduction in estimates per CPU time.

For inverse problems, applying standard MH can be computationally costly as the cost  is dominated by the evaluation of the posterior density, which involves simulating of the forward map $F({\bf x}')$ at proposed parameters.  
Instead, we consider using an approximate posterior, obtained by using some reduced model of the forward model, to improve the computational efficiency. 
We embed the approximation in the DA algorithm shown in Alg.~\ref{alg:da} to obtain asymptotically unbiased MCMC estimates.

\begin{algorithm}[h]
\caption{Delayed Acceptance (DA)}\mbox{} 
\label{alg:da}
At iteration $n$, given ${\bf x}_n = {\bf x}$ and a (possibly state-dependent) approximation to the target distribution $\pi_{\bf x}^*(\,\cdot\,|\,\tilde{\bf d})$, then ${\bf x}_{n+1}$ is determined in the following way:
\begin{enumerate}
\item \label{step:da1} \label{step:da2} Propose new state ${\bf y}$ from some distribution $q\left({\bf x}, \,\cdot\,\right)$. With probability 
\begin{equation*}
  \alpha\left({\bf x},{\bf y}\right) = \min\left\{1, \frac{\pi_{\bf x}^*\left({\bf y}\,|\,\tilde{\bf d}\right) q\left({\bf y} ,{\bf x}\right)}
                     {\pi_{\bf x}^*\left({\bf x}\,|\,\tilde{\bf d}\right)q\left({\bf x} ,{\bf y}\right)}   ,
      \right\}
\end{equation*}
promote ${\bf y}$ to be used in Step 2, otherwise set ${\bf y}={\bf x}$ (or, equivalently, set ${\bf x}_{n+1} = {\bf x}$ and exit).
\item \label{step:da3} The effective proposal distribution at the second step is
\begin{equation*}
  q^*\left({\bf x} ,{\bf y}\right) = 
    q\left({\bf x} ,{\bf y}\right)\,\alpha\left({\bf x},{\bf y}\right) 
    + \left[ 1-\int q\left({\bf x} ,{\bf z}\right)\,\alpha\left({\bf x},{\bf z}\right)\,\dd{\bf z}\right]\delta_{\bf x}({\bf y}).
\end{equation*}
With probability 
\begin{equation*}
  \beta\left({\bf x},{\bf y}\right) =\min\left\{1, \frac{\pi_{\rm post}\left({\bf y} \, | \, \tilde{\bf d}\right) q^*\left({\bf y} ,{\bf x}\right)}
                     {\pi_{\rm post}\left({\bf x} \, | \, \tilde{\bf d}\right)q^*\left({\bf x} ,{\bf y}\right)}   
      \right\}
\end{equation*}
set ${\bf x}_{n+1} = {\bf y}$, otherwise ${\bf x}_{n+1} = {\bf x}$. 
\end{enumerate}
\end{algorithm} 
The DA algorithm allows the approximation to depend on the state of the MCMC. Hence DA generalizes the surrogate transition method of~\cite{mcmc_liu} that uses a fixed approximate target distribution, or `surrogate'. 

DA uses two accept-reject steps. The first step uses an approximation to the target distribution, while the second accept-reject step ensures that the Markov chain correctly targets the desired distribution (see~\cite{ChristenFox2005} for details).
The resulting  transition kernel is
\begin{equation} 
K^*({\bf x, y}) = q( {\bf x, y} )\,\alpha( {\bf x, y} )\,\beta( {\bf x, y} ) + \left[ 1 - \int q( {\bf x, z} )\,\alpha( {\bf x, z} )\,\beta( {\bf x, z} ) d {\bf z} \right] \delta_{\bf x}({\bf y}),
\end{equation} 
that composes the proposal density, the first step acceptance probability, and the second step acceptance probability.
Given a proposal $q\left({\bf x}, \,\cdot\,\right)$ for which a single level MH satisfies reversibility, recurrence, and irreducibility, DA is ergodic under mild conditions and has $\pi_{\rm post}( \, \cdot\,| \, \tilde{\bf d})$ as the unique stationary distribution. \cite{ChristenFox2005}

In Alg.~\ref{alg:da}, there is never need to evaluate the integral in the definition of $q^*$. The computational cost per iteration is reduced because only those proposals that are accepted using the approximation  $\pi_{\bf x}^*(\,\cdot\,|\,\tilde{\bf d})$ go on to evaluation of the posterior distribution $\pi_{\rm post}(\,\cdot\, | \,\bf \tilde{d})$, that requires evaluating the full, expensive forward map.
However, DA necessarily has lower statistical efficiency than the unmodified counterpart in Alg.~\ref{alg:mh}, because of the additional acceptance/rejection in Step~\ref{step:da3}~\cite{ChristenFox2005}. That is, $\tau_{\mathrm{DA}}\geq\tau$ for any quantity $g$.
Fortunately, DA may still be more computationally efficient that the standard MH.

Here we analyze the potential speed-up factor in computational efficiency of DA compared to the standard MH. 
Let the CPU time to evaluate the approximate posterior density and the exact posterior density be $t^*$ and $t$, respectively. Suppose that the average acceptance probability in Step~\ref{step:da2} of DA is $\bar{\alpha}$, set by the choice of proposal.
Since evaluation of the posterior density dominates the CPU time of MH, simulating $n$ iterations of the standard MH asymptotically costs $n\, t$ CPU time.
Using the same CPU time, DA can be simulated for $n\, t / (\bar{\alpha} t + t^* )$ iterations. This way, the effective sample size of DA is given by 
\[
  \textrm{ESS}_{\rm DA} = \frac{n\, t}{\tau_{\rm DA} \, (\bar{\alpha} t + t^* ) } = \frac{n}{\tau_{\rm DA} \, (\bar{\alpha} + t^*/t ) } ,
\] 
whereas the effective sample size of the standard MH is 
\[
  \textrm{ESS} = \frac{n}{\tau }  .
\] 
Thus, the speed-up factor of DA compared to the standard MH is 
\begin{equation} 
	\frac{\mbox{ESS}_{\mathrm{DA}}}{\mbox{ESS}} = 
	\frac{\tau}{\tau_{\mathrm{DA}}}\frac{1}{\bar{\alpha} + t^* / t}.
	\label{eq:su_factor1}
\end{equation} 
The factor $\tau/\tau_{\mathrm{DA}} \leq 1$ counts the decrease in statistical efficiency by using DA instead of MH, while the factor $\bar{\alpha} + t^* / t$ gives the decrease in average compute cost per iteration. 
It is necessary to address both factors if one is to improve computational efficiency. 
That is, we want $\tau/\tau_{\mathrm{DA}}$ to be close to one and $t^* / t$ to be as small as possible.  

Several forms of reduced models have been used in DA to approximate the posterior. For example, local linearized models have been used in~\cite{ChristenFox2005} and model coarsening based on multiscale finite element is used in~\cite{Efendiev}.
It can be challenging to balance the reduction in CPU time against accuracy of the reduced model. 
Using a lower accuracy reduced model, which runs faster compared to a more accurate one, will reduce the CPU time of MCMC per iteration, but at the risk of lower statistical efficiency since DA is more likely to reject the proposal in Step~\ref{step:da3}.

Here we present a systematic way to modify the statistical model of the likelihood function without changing the reduced model.
By including the statistics of the numerical error of the reduced model in the likelihood, our corrected approximate posterior can  potentially improve the acceptance rate in step~\ref{step:da3}, at no extra cost.
This way, the statistical efficiency, and hence the speed-up factor, of the DA can be improved by using the same reduced model.

\subsection{Adaptive MCMC}
\label{sec:adapt}
A key ingredient of our modified likelihood is the use of adaptive MCMC to estimate posterior error statistics of the reduced model.  
Towards this goal, we design an adaptive delayed acceptance algorithm that both adaptively adjusts the proposal distribution and also adapts to the posterior error statistics in the likelihood using the MCMC sample history.
The adaptation of the likelihood is significantly different from existing adaptive MCMC algorithms which only focus on adapting the proposal. 
This offers new insights in the adaptive construction of goal-oriented approximations to the likelihood for posterior exploration. 
In this paper, we extend the regularity conditions  established by \cite{adaptive_roberts2007,LRR2013} to our case, to guarantee the ergodicity of our adaptive scheme. 

We first review classical adaptive MCMC methods that focuses on tuning proposal distributions using past MCMC samples. 
This includes the classical adaptive Metropolis (AM) algorithm \cite{adaptive_haario2001} and our modification that suits our case studies. 

\subsubsection{Adaptive Metropolis}

An important task in practical MCMC is identifying good proposal distributions that can minimize autocorrelation. As the forward model used in our case studies does not have adjoint capabilities, advanced MCMC proposals, such as the stochastic Newton \cite{MCMC:MWBG_2012} and the likelihood-informed dimension-independent proposal \cite{MCMC:CLM_2016} that rely on derivatives of the posterior, cannot be used here. We consider the Gaussian random walk proposal
\[
q({\bf x},\cdot) = \N({\bf x},\sigma^2{\boldsymbol\Sigma}),
\] 
where $\boldsymbol\Sigma$ is some covariance matrix and $\sigma$ is a scalar that dictates the jump size.
Under certain technical assumptions, the covariance of the posterior can be a good choice for $\boldsymbol\Sigma$. 
There are two ways of choosing an optimal scaling $\sigma$, see Roberts and Rosenthal \cite{mcmc_roberts2001} and references therein for detailed analysis. 
One is to set $\sigma = 2.38/\sqrt{d}$, where $d$ is the dimension of the parameter. 
Alternatively, one can choose an scaling $\sigma$ such that the acceptance rate of MCMC is about $0.234$.
The values $2.38$ and $0.234$ are numerical approximations to analytical results.
These two choices are equivalent under strict technical assumptions \cite{mcmc_roberts2001}, however the latter often demonstrates better efficiency in practice. 

Since the posterior covariance is unknown before carrying posterior sampling, Harrio et al. \cite{adaptive_haario2001} introduced the adaptive Metropolis that estimates the posterior covariance on-the-fly using past posterior samples generated during the MCMC simulation. 
%
%
AM uses $\boldsymbol\Sigma\approx\boldsymbol\Sigma_n$ where $\boldsymbol\Sigma_n$ is the empirical posterior covariance evaluated over $n$ iterations, with the scale is chosen as $\sigma = 2.38/\sqrt{d}$. 
The scaled empirical posterior covariance is mixed with a fixed Gaussian distribution, $\N({\bf x},(0.1^2/d){\bf I}_d)$ to avoid the situation that  $\boldsymbol\Sigma_n$ is singular.
The algorithm is shown in Alg.~\ref{alg:am}. 

\begin{algorithm}[h!]
\caption{Adaptive Metropolis}
\label{alg:am}
At iteration $n$, given ${\bf x}_n = {\bf x}$, then ${\bf x}_{n+1}$ is determined in the following way:
\begin{enumerate}
\item 
Given the empirical covariance estimate $\Sigma_n$ of the target distribution up to step $n$ and a small positive constant $\beta$, propose a new state ${\bf y}$ from the proposal
\begin{equation} 
q_n({\bf x},\cdot) = \left\{ \begin{array}{ll} \N\big({\bf x},\frac{0.1^2}{d} {\bf I}_d\big) & n \leq 2d \\ \N\big({\bf x},(1-\beta)\frac{2.38^2}{d}\boldsymbol\Sigma_n+\beta \frac{0.1^2}{d}{\bf I}_d\big) & n > 2d \end{array} \right.,
\label{eq:am}
\end{equation} 
where $d$ is the dimension of the parameter.
\item With probability $\min\big\{1, \pi_{\rm post}({\bf y} \, | \, \tilde{\bf d}) / \pi_{\rm post}({\bf x} \, | \, \tilde{\bf d}) \big\}$, ${\bf x}_{n+1} = {\bf y}$, otherwise ${\bf x}_{n+1} = {\bf x}$.
\end{enumerate}
\end{algorithm}

Note that the proposal is modified in every iteration of adaptive MCMC sampling. Thus, it becomes non-trivial to establish ergodicity of the resulting algorithm. For a non-adaptive proposal the transition kernel of MH satisfies the detailed balance condition
\[
\pi_{\rm post}({\bf x}_n \, | \, \tilde{\bf d})  K({\bf x}_{n}, {\bf x}_{n+1}) = \pi_{\rm post}({\bf x}_{n+1} \, | \, \tilde{\bf d}) K({\bf x}_{n+1}, {\bf x}_{n}) ,
\]
and hence reversibility.
For an adaptive proposal, detailed balance is clearly violated as the proposal and the transition kernel depend on the iteration number $n$.
%
%
Results by \cite{adaptive_haario2001,adaptive_andr2006,adaptive_roberts2007,adaptive_atch2005,LRR2013} and others have established ergodicity for some adaptive MCMC proposals using differing techniques.
Roberts and Rosenthal \cite{adaptive_roberts2007} provided simplified regularity conditions required for ergodicity, namely simultaneous ergodicity and diminishing adaptation, that provides a viable route to establishing ergodicity of many adaptive MCMC algorithms, and is the route we take here.

\subsubsection{Grouped components adaptive Metropolis}
\label{sec:gcam}
As demonstrated by a set of numerical examples \cite{phd_cui}, the scaling $2.38/\sqrt{d}$ often shows suboptimal statistical performance, as the technical assumptions used to derive the scale are typically too restrictive for highly non-Gaussian posterior distributions with correlated parameters. 
In addition, it may be computationally costly to estimate the covariance $\Sigma$ for problems with high parameter dimensions.
To overcome these limitations, we present the grouped components adaptive Metropolis (GCAM) proposal that separately estimates the covariance $\Sigma$ and the scale $\sigma$ on different groups in a partition of the parameter coordinates.
%
%
Suppose we have $L$ groups of components $\mathcal{I}_1,\ldots,\mathcal{I}_L$, $\mathcal{I}_j \subseteq \{1,2,\ldots,d\}$ with each group $\mathcal{I}_j$ associated with a scale variable $\sigma_{j}$. Let $d_j$ be the number of elements of group $\mathcal{I}_j$, and $-\mathcal{I}_j = \{1,2,\ldots,d\} \setminus \mathcal{I}_j$. 
The GCAM proposal is shown in Alg. \ref{alg:gcam}.

\begin{algorithm}[h!]
\caption{Grouped Components Adaptive Metropolis}
\label{alg:gcam}
At iteration $n$, given ${\bf x}_n = {\bf x}$, then ${\bf x}_{n+1}$ is determined in the following way:
\begin{enumerate}
\item Initialize $\bf y = x$, for all $j = 1\ldots,L$:
\begin{itemize}
  \item Given the empirical covariance $\boldsymbol\Sigma_{n,\mathcal{I}_j}$ for the components $\mathcal{I}_j$ estimated from past samples and a small positive constant $\beta$, draw a $d_j$ dimensional random variable $\bf z$ from the proposal
				\begin{equation} 
					q_n({\bf x}_{\mathcal{I}_j},\cdot) = \left\{ \begin{array}{ll} \N\big({\bf x}_{\mathcal{I}_j},\frac{0.1^2}{d_j} {\bf I}_{d_j}\big) & n \leq 2d_j \\ \N\Big({\bf 		x}_{\mathcal{I}_j},\frac{\sigma_{j}^2}{\max_{i\in\mathcal{I}_j}\{\boldsymbol\Sigma_{n,\mathcal{I}_j}(i,i)\} }\big(\boldsymbol\Sigma_{n,\mathcal{I}_j} + \beta{\bf I}_{d_j}\big)\Big) & n > 2d_j		\end{array} \right..
					\label{eq:gcam}
				\end{equation} 
	\item With probability $\min\left\{ 1, \pi_{\rm post}({\bf y}_{-\mathcal{I}_j}, {\bf z}\,|\,{\bf \tilde{d}}) / \pi_{\rm post}({\bf y} \, | \, \tilde{\bf d})  \right\}$, set ${\bf y}_{\mathcal{I}_j} = \bf z$, otherwise ${\bf y}_{\mathcal{I}_j}$ unchanged.
\end{itemize}
\item Then ${\bf x}_{n+1} = \bf y$ after updating all the $L$ groups of components.
\item For a pre-specified batch number $N$, if $n \bmod N = 0$, for all $j = 1\ldots,L$:
\begin{itemize}
	\item Calculate the acceptance rate $\hat{\alpha}_j$ from the past $N$ updates on $j$th group of components,
	\item If $\hat{\alpha}_j > 0.234$, $\sigma_{j} = \sigma_{j} \,\exp(\delta)$, otherwise $\sigma_{j} = \sigma_{j} \, \exp(-\delta)$.
\end{itemize}
Here, $\delta = \min\{0.01, \sqrt{N/n}\}$.
\end{enumerate}
\end{algorithm} 
%

In the proposal distribution~\eqref{eq:gcam}, $1/\max_{i\in\mathcal{I}_j}\{\boldsymbol\Sigma_{n,\mathcal{I}_j}(i,i)\}$ gives the inverse of the largest variance among parameters in the group $\mathcal{I}_j$. This factor approximately normalizes the covariance matrix $\boldsymbol\Sigma_{n,\mathcal{I}_j}$, and hence avoids the interaction of $\boldsymbol\Sigma_{n,\mathcal{I}_j}$ and $\sigma_j$ during the adaptation. 
Such numerical treatment makes the scale variable $\sigma_j$ stabilize faster. 
The ergodicity of GCAM can be shown using the simplified conditions of Roberts and Rosenthal. \cite{adaptive_roberts2007} We will discuss this further in Section 4.
%

%
\section{Approximations to the Forward Map and Posterior Distribution}
\label{sec:approx}

Apart from the high dimensionality and complex nature of the posterior, a significant computational challenge arises from the high computational cost of evaluating the posterior density that entails computationally demanding numerical schemes used by the forward model $F(\cdot)$.
For example, the 3D geothermal reservoir model presented in Section~\ref{sec:geo} has about ten thousand parameters, and each model evaluation takes about 30 to 50 minutes CPU time. 

To improve the computational efficiency of MCMC as discussed previously, we approximate the likelihood, and hence the posterior, by employing reduced models, denoted by $F^*(\cdot)$. Efficiency of the MCMC requires the reduced model to be accurate and cheap, though these requirements are application specific. 
The backbone of the approximation developed here is a given reduced model built using existing techniques, including grid coarsening~\cite{upscaling_christie_2001,model_kaipio1,Efendiev,phd_cui}, linearization of the forward model~\cite{ChristenFox2005}, and projection-based methods~\cite{mor_BGW_2015, MOR_DEIM_reservoir, mor_LWG_2015, mor_CMW_2015, mor_APL_2016}.
Our key contribution here is to present a new way to improve the approximation to the likelihood function by considering the posterior statistics of the numerical error of the reduced model.
This leads to the adaptive delayed acceptance algorithm with substantial improvement in the computational efficiency compared to the classical delayed acceptance or surrogate transition algorithms.

In this section, we first present posterior approximation using reduced models and the classical approximation error models, then discuss various ways to estimate posterior error statistics and correct the approximated posteriors. 


\subsection{Approximation Error Models}
\label{sec:aem}
As discussed above, there are many possible ways to build the reduced model $F^*(\cdot)$. 
The only restriction we impose on $F^*(\cdot)$ is the condition required for DA to be valid~\cite{ChristenFox2005}, essentially that  $F^*(\cdot)$ is finite valued, which is virtually no restriction in practice. 
We start with the common approximation to the posterior that simply uses $F^*(\cdot)$ in place of the true forward map $F(\cdot)$.
\begin{appr}
\label{appr:coarse}
Approximate posterior distribution using $F^*(\cdot)$ in place of $F(\cdot)$, which has the probability density function
\begin{equation} 
\pi_{\rm post}^*({\bf x} \, | \, \tilde{\bf d}) \propto \exp\left[ -\frac{1}{2}  \big( F^*({\bf x})-{\bf \tilde{d}} \big)^\top \boldsymbol\Sigma_{\bf e}^{-1} \big( F^*({\bf x})-{\bf \tilde{d}} \big) \right] \pi_{\rm prior}({\bf x}).
\label{eq:coarse_post}
\end{equation} 
\end{appr}
Using Approximation~\ref{appr:coarse} in DA, Alg.~\ref{alg:da}, gives the surrogate transition method~\cite{mcmc_liu,Efendiev}.

The reduced model usually has a non-negligible discrepancy with the forward model, and hence the approximation in Eqn.~\eqref{eq:coarse_post}, by itself,  can result in biased estimates while producing uncertainty intervals that are too small. This effect was investigated by~\cite{stats_inverse,model_kaipio1}, who noted that this is one way to perform an ``inverse crime''. This indicates that the approximate posterior has displaced support and is too narrow to include the support of the true posterior. 

One possible remedy is afforded by considering the statistics of the numerical error of the reduced model. Given a reduced model $F^*(\cdot)$, Eqn.~\eqref{eq:cali} can be expressed as
\begin{eqnarray}
{\bf \tilde{d}} & = & F^*({\bf x}) + \big( F({\bf x}) - F^*(\bf x) \big) + \bf e\nonumber \\
& =  & F^*({\bf x}) + B({\bf x}) + \bf e,
\label{eq:enhanm}
\end{eqnarray} 
where  $B({\bf x})$ is the model reduction error between the true forward model and the reduced model. Assuming that the model reduction error is independent of the model parameters and normally distributed gives the approximation error model (AEM)~\cite{model_kaipio1}
\[
{\bf \tilde{d}}  =  F^*({\bf x}) + B + \bf e,
\] 
where $B \sim N({\boldsymbol \mu}_B, \boldsymbol\Sigma_B)$. Kaipio \& Somersalo~\cite{model_kaipio1} showed that this improved the approximation of the posterior distribution, compared to Approximation~\ref{appr:coarse}, in their applications, using an {\it a priori} construction of the AEM. That is, before utilizing data and solving the inverse problem, the AEM was empirically estimated over the prior distribution. The resulting estimates for mean and covariance of $B$ are
\begin{eqnarray}
\mu_B & = & \int_{\mathcal{X}} B({\bf x}) \pi_{\rm prior}({\bf x}) d{\bf x} 
       \approx  \frac{1}{L}\sum_{i = 1}^{L} B({\bf x}_i),  \label{eq:muprior}\\
\Sigma_B & = & \int_{\mathcal{X}} \big( B({\bf x}) - \mu_B \big) \, \big( B({\bf x}) - \mu_B \big)^\top \pi_{\rm prior}({\bf x}) d{\bf x} 
       \approx  \frac{1}{L-1}\sum_{i = 1}^{L} \big( B({\bf x}_i) - \mu_B \big) \,\big( B({\bf x}_i) - \mu_B \big)^\top, \label{eq:sigmaprior}
\end{eqnarray} 
where $B({\bf x})$ is defined by Eqn.~\ref{eq:enhanm}, and ${\bf x}_i \sim \pi_{\rm prior}({\bf x}), i = 1,\cdots,L$, are $L$ samples drawn from the prior distribution.
Using ${\boldsymbol \mu}_B$ and $\boldsymbol\Sigma_B$ estimated from the prior distribution, we have the following approximate posterior.

\begin{appr}
\label{appr:aempr}
Approximate posterior distribution using the reduced model $F^*(\cdot)$ and AEM estimated from the prior, which has the probability density function
\begin{equation} 
\pi_{\rm post}^*({\bf x} \, | \, \tilde{\bf d}) \propto \exp\left[ -\frac{1}{2}  \big( F^*({\bf x})+\boldsymbol\mu_{B}-{\bf \tilde{d}} \big)^\top \big(\boldsymbol\Sigma_{B}+\boldsymbol\Sigma_{\bf e}\big)^{-1} \big( F^*({\bf x})+\boldsymbol\mu_{B}-{\bf \tilde{d}} \big) \right] \pi_{\rm prior}({\bf x}).
\label{eq:a_enhan_prior}
\end{equation} 
\end{appr}

\subsection{Posterior Approximation Error Models}
\label{sec:postAEM}
There are two major drawbacks when using the AEM estimated from the prior. 
Firstly, 
the support of the prior is typically quite different to the support of the posterior, as the observed data in the likelihood necessarily make the posterior concentrate compared to the prior. 
Hence, the AEM may be reasonable over the prior distribution but the $L$ samples may not include any samples with appreciable posterior probability so the AEM could be far from optimal over the support of the posterior distribution.  
Secondly, the {\it a priori} approach requires appreciable pre-computation to construct the AEM.

We improve the AEM, and hence make a better approximate posterior distribution, by empirically estimating the AEM over the posterior distribution. That is, we estimate the mean and covariance of the AEM by 
\begin{eqnarray}
\mu_B & = & \int_{\mathcal{X}} B({\bf x}) \pi_{\rm post}({\bf x} \, |\, \tilde{\bf d}) d{\bf x} , \label{eq:postmu} \\
\Sigma_B & = & \int_{\mathcal{X}} \big(B({\bf x}) - \mu_B\big)\,\big(B({\bf x}) - \mu_B\big)^\top \pi_{\rm post}({\bf x} \, |\, \tilde{\bf d}) d{\bf x}. \label{eq:postsigma} 
\end{eqnarray} 

For each accepted posterior sample of DA (Alg.~\ref{alg:da}), we can calculate the model error between full model and the reduced model, since both are evaluated at each acceptance, and define
$B_{{\bf x}_{n-1}}({\bf x}_{n})=B({\bf x}_{n})-B({\bf x}_{n-1})$.
Since DA can sample the full posterior, we can adaptively estimate $\mu_B$ and $\Sigma_B$ from the posterior samples during the simulation of DA.
A carefully designed adaptive MCMC can make the \textit{a posteriori} estimates of $\mu_B$ and $\Sigma_B$ converge to the true values given in Eqns~\ref{eq:postmu} and \ref{eq:postsigma}. 
%
%
Let $n$ denote the MCMC iteration, the estimated mean and covariance of AEM denoted by $\bar{\boldsymbol\mu}_{B,n}$ and $\bar{\boldsymbol\Sigma}_{B,n}$, respectively, can be iteratively updated as
\begin{eqnarray}
  \label{eq:updatemu1}
  \bar{\boldsymbol\mu}_{B,n}  & = & \frac{1}{n}\big[(n-1)\bar{\boldsymbol\mu}_{B,n-1} + B_{{\bf x}_{n-1}}({\bf x}_n)\big],\\
  	\bar{\boldsymbol\Sigma}_{B,n} & = &  \frac{1}{n-1} \big[ (n-2)\bar{\boldsymbol\Sigma}_{B,n-1}  + B_{{\bf x}_{n-1}}({\bf x}_{n})B_{{\bf x}_{n-1}}({\bf x}_{n})^\top\big].
  \label{eq:updatesigma1}
\end{eqnarray} 
This defines an approximate posterior that changes adaptively over MCMC simulations.
\begin{appr}
\label{appr:aempost}
Adaptive approximate posterior distribution using the reduced model $F^*(\cdot)$ and AEM adaptively estimated over the posterior, giving the posterior probability density function at iteration $n$
\begin{equation} 
\pi_{n,\rm post}^*({\bf x} \, | \, \tilde{\bf d}) \propto \exp\left[ -\frac{1}{2}  \big( F^*({\bf x})+\bar{\boldsymbol\mu}_{B, n}-{\bf \tilde{d}} \big)^\top \big(\bar{\boldsymbol\Sigma}_{B, n}+\boldsymbol\Sigma_{\bf e}\big)^{-1} \big( F^*({\bf x})+\bar{\boldsymbol\mu}_{B,n}-{\bf \tilde{d}} \big) \right] \pi_{\rm prior}({\bf x}).
\label{eq:a_enhan_post}
\end{equation} 
\end{appr}

We can use the adaptive posterior in Approximation~\ref{appr:aempost} within DA to speed-up MCMC sampling, while using the resulting posterior samples to simultaneously update the AEM in the approximation. 
%
%
A further advantage of this adaptive approach is that it does not require any pre-computation to estimate  the AEM before setting up a MCMC simulation. Indeed, in all computational experiments we find that building the AEM over the posterior leads to better statistical efficiency in the MCMC, than when the AEM is estimated over the prior, and so gives a method that both does not require precomputation and also gives a more computationally efficient MCMC.

A critical issue arises because the approximate posterior changes over iterations in the MCMC simulation. As with the adaptive Metropolis case, it becomes unclear whether DA using this adaptive approximation is ergodic. We extend the current framework for adaptive MCMC~\cite{adaptive_roberts2007,adaptive_haario2001}, in Section 4, to establish ergodicity of the resulting DA scheme using adaptive approximate posteriors. 

\subsection{State-dependent Approximation Error Models}

Christen \& Fox~\cite{ChristenFox2005} demonstrated DA using a local linearization of the forward map as the approximate forward map, which is a local reduced model that depends on the current state ${\bf x}$ of the MCMC. For the applications we present in Section~\ref{sec:geo}, we utilize the existing Fortran code TOUGH2 \cite{tough2} to simulate the forward map. This package does not give access to derivatives (nor adjoints, etc) and so we form a reduced model $F^*(\cdot)$ by using a coarsened discretization. This reduced model  depends only on the point ${\bf y}$ at which it is evaluated, but not the state of the MCMC.

In many cases, including in the applications we consider here, both the true forward map and reduced model are $P$-H\"{o}lder continuous, i.e., $\|F({\bf y})-F({\bf x})\| \leq C\|{\bf y-x}\|^P$ for some $C>0$ and $P>0$. Then, when using DA, a local improvement to the state-independent reduced model can be made, for little additional computational cost, by using the values of $F({\bf y})$ and $F^*({\bf y})$ for points  ${\bf y}$ that are accepted, and hence become the state of the chain.
We define a deterministic state-dependent reduced model using a zeroth-order correction to the reduced model, as follows.
\begin{appr}
\label{appr:local}
State-dependent reduced model and approximate posterior distribution: Suppose that at iteration $n$, the Markov chain has state ${\bf x}_n = {\bf x}$. For a proposed state ${\bf y} \sim q({\bf x},\cdot)$, the state-dependent reduced model  $F^*_{\bf x}(\cdot)$ is 
\begin{equation} 
F^*_{\bf x}({\bf y}) = F^*({\bf y}) + \big(F({\bf x}) - F^*({\bf x})\big).
\label{eq:rom_l}
\end{equation} 
The resulting approximate posterior density function is
\begin{equation} 
\pi^*_{{\bf x}}({\bf y} \, | \, \tilde{\bf d})  \propto \exp\left[-\frac{1}{2}  \big( F^*_{\bf x}({\bf y})-{\bf \tilde{d}}\big)^\top \boldsymbol\Sigma_{\bf e}^{-1}\big(F^*_{\bf x}({\bf y})-{\bf \tilde{d}}\big)\right] \pi_{\rm prior}({\bf y}).
\end{equation} 
\end{appr}
It is worth mentioning that the state-dependent reduced model~\eqref{eq:rom_l} has the desirable property that $F^*_{\bf x}({\bf x})=F({\bf x})$. 
The error structure of the state-dependent reduced model~\eqref{eq:rom_l} can also be estimated by employing the AEM. In particular,  Approximation~\ref{appr:aempost} and \ref{appr:local} can be combined, at no significant increase in computational cost, to give a more accurate approximation to the posterior distribution. The state-dependent model reduction error, for the zeroth-order correction \eqref{eq:rom_l}, is 
\[
B^{}_{\bf x}({\bf y}) = F({\bf y}) - F^*_{\bf x}({\bf y}),
\]  
where $B_{\bf x}({\bf x}) = {\bf 0}$.
This way, the mean and covariance of the AEM in Approximation \ref{appr:aemsd} with reduced model \eqref{eq:rom_l} are  
\begin{eqnarray}  
\label{eq:meanAAEM}
  \boldsymbol\mu_{B} &=&\quad\textrm{E}_{\pi_{\rm post}}\Big[\int_{\mathcal{X}} B_{\bf x}({\bf y}) K({\bf x, y}) d{\bf y}\Big], \\
   \boldsymbol\Sigma_{B}&=&\textrm{Cov}_{\pi_{\rm post}}\Big[\int_{\mathcal{X}} B_{\bf x}({\bf y}) K({\bf x, y}) d{\bf y}\Big],
   \label{eq:covAAEM}
\end{eqnarray} 
where $K({\bf x, y})$ is the transition kernel implemented by the MCMC iteration.
The mean of the AEM~\eqref{eq:meanAAEM} for reduced model \eqref{eq:rom_l} can be shown to be ${\bf 0}$, as follows.
\begin{eqnarray*}
  \E_{\pi_{\rm post}}\Big[\int B_{\bf x}({\bf y}) K({\bf x, y}) d{\bf y}\Big]  
  & 
  = \int_{\mathcal{X}} \int_{\mathcal{X}} \big(F({\bf y}) - F^*({\bf y})\big) \pi_{\rm post}({\bf x}\, | \, \tilde{\bf d}) K({\bf x, y}) d{\bf y} d{\bf x} \\ 
  & 
  - \int_{\mathcal{X}} \int_{\mathcal{X}} \big(F({\bf x}) - F^*({\bf x})\big) \pi_{\rm post}({\bf x}\, | \, \tilde{\bf d}) K({\bf x, y}) d{\bf y} d{\bf x}
\end{eqnarray*} 
with the two terms on the right canceling because the kernel $K$ satisfies the detailed balance condition $\pi_{\rm post}({\bf x}\, | \, \tilde{\bf d}) K({\bf x, y}) =  \pi_{\rm post}({\bf y}\, | \, \tilde{\bf d}) K({\bf y, x})$. 
%
%
Since we have $\boldsymbol\mu_{B}={\bf 0}$, the covariance of the model reduction error $B_{\bf x}({\bf y})$ can be computed adaptively at iteration $n$ by the inductive formula
\begin{equation} 	
	\hat{\boldsymbol\Sigma}_{B,n} =  \frac{1}{n-1} \left[ (n-2)\hat{\boldsymbol\Sigma}_{B,n-1}  + B_{{\bf x}_{n-1}}({\bf x}_{n})B_{{\bf x}_{n-1}}({\bf x}_{n})^\top\right].
	\label{eq:appr41}
\end{equation} 
The approximate posterior distribution based on state-dependent reduced model \ref{eq:rom_l} and AEM estimated from the posterior takes the following form.
\begin{appr}
\label{appr:aemsd}
AEM built over the posterior distribution with a state-dependent reduced model: Suppose that at iteration $n$ the Markov chain is at state ${\bf x}_n = {\bf x}$ and a proposed state is ${\bf y} \sim q({\bf x},\cdot)$.  The state-dependent approximate posterior density function is given by
\begin{equation} 
\pi^*_{n,\bf x}({\bf y} \, | \, \tilde{\bf d}) \propto  \exp\left[-\frac{1}{2}  \big(F^*_{\bf x}({\bf y})-{\bf \tilde{d}}\big)^\top \big(\hat{\boldsymbol\Sigma}_{B,n}+\boldsymbol\Sigma_{\bf e}\big)^{-1} \big(F^*_{\bf x}({\bf y})-{\bf \tilde{d}}\big)\right] \pi_{\rm prior}({\bf y}).
\label{eq:a_enhan_lp}
\end{equation} 
\end{appr}
As for Approximation~\ref{appr:aempost} we have to show ergodicity of the resulting adaptive MCMC scheme.

\section{Adaptive Delayed Acceptance Algorithm}
\label{sec:adamh}
The ADA algorithm combines the adaptive AEM built over the posterior, suggested in Section~\ref{sec:approx}, and the adaptation of the proposal distribution as in existing adaptive algorithms in Section~\ref{sec:adapt}. Adaptive estimates of the AEM are made possible by using the DA algorithm, which provides the basic structure of ADA, and also the mechanism for reduction in compute cost per iteration. The use of adaptively improved stochastic approximations means the reduction in statistical efficiency, which is an unavoidable consequence of using DA, may be minimized so the reduction of compute cost per iteration is transferred directly into computational efficiency of the resulting algorithm.

\begin{algorithm}[h!]
\caption{Adaptive delayed acceptance Metropolis-Hastings (ADA)}
\label{al:adamh} 
At iteration $n$, given ${\bf x}_n = {\bf x}$, adaptive proposal $q_n({\bf x}, \,\cdot\,)$, and approximate posterior distribution $\pi^*_{n,{\bf x}}(\,\cdot\, | \,\tilde{\bf d})$, then ${\bf x}_{n+1}$ and updated distributions are determined as follows:
\begin{enumerate}
\item Generate a proposal ${\bf y}\sim q_n({\bf x}, \,\cdot\,)$. \label{step:adamh1} \label{step:adamh2} With probability
\begin{equation*}
\alpha_n({\bf x} ,{\bf y}) = \min\left\{ 1, \frac{\pi^*_{n,{\bf x}}({\bf y}\, | \,\tilde{\bf d}) \, q_n({\bf y}, \bf x)}{ \pi^*_{n,{\bf x}}({\bf x}\, | \,\tilde{\bf d}) \, q_n({\bf x}, \bf y) }\right\} 
\end{equation*}
 promote ${\bf y}$ to be used as a proposal for the following step. Otherwise set ${\bf y} ={\bf x}$ and proceed. 
\item
The proposal distribution at this step is
\begin{equation*}
q^*_n({\bf x} ,{\bf y}) = \alpha_n({\bf x} ,{\bf y})q_n({\bf x}, {\bf y}) + \Big[1-\int_{\mathcal{X}} \alpha_n({\bf x} ,{\bf y})q_n({\bf x}, {\bf y}) \,\dd{\bf y}\Big]\delta_{\bf x}({\bf y}),
\end{equation*}
where $\delta_{\bf x}(\cdot)$ denotes the Dirac mass at ${\bf x}$. 
With probability
\begin{equation*}
 \min\left\{ 1, \frac{\pi_{\rm post}({\bf y} \, | \, \tilde{\bf d})q^*_n({\bf y} ,{\bf x})}{\pi_{\rm post}({\bf x} \, | \, \tilde{\bf d})q^*_n({\bf x},{\bf y})} \right\} 
\end{equation*} 
set ${\bf x}_{n+1}= {\bf y}$.  Otherwise set ${\bf x}_{n+1}=\bf x$.
\label{step:adamh3}
\item Form the updated approximation $\pi^*_{n+1,{\bf x}}(\,\cdot\, | \,\tilde{\bf d})$. \label{step:adamh4}
\item Form the updated adaptive proposal $q_{n+1}({\bf x}, \,\cdot\,)$. \label{step:adamh5}
\end{enumerate}
Here we use the general notation $\pi^*_{n,{\bf x}}(\,\cdot\, | \,\tilde{\bf d})$ to denote the approximate posterior. It can be either state-dependent or state-independent, either adaptive or non-adaptive.
\end{algorithm}

We summarize ADA in Alg.~\ref{al:adamh}. In this algorithm, the proposal $q_n( \cdot , \cdot)$ and its adaptive update in Step~\ref{step:adamh5} may have the form of any of the classical adaptive MCMC algorithms, such as the AM in Alg. \ref{alg:am} and the GCAM in Alg. \ref{alg:gcam}.
We use the general notation $\pi^*_{n,{\bf x}}(\,\cdot\, | \,\tilde{\bf d})$ to denote the approximate posterior. It includes the state-dependent, state-independent, adaptive, and non-adaptive cases.
When the adaptive approximate posteriors \ref{appr:aempost} and \ref{appr:aemsd} are used, the corresponding adaptive error models have to be updated in Step~\ref{step:adamh4} according to Eqn~\ref{eq:updatemu1} and \ref{eq:updatesigma1} and Eqn. \ref{eq:appr41}, respectively.

We note that ADA in Alg. \ref{al:adamh} is not restricted to the form of the approximate posteriors used here. It offers a general framework for constructing other forms of posterior approximations in Step 3 using the forward model evaluated at past posterior samples.
In the rest of this section, we present regularity conditions on the adaptive approximation and adaptive proposal that guarantee ergodicity of ADA. These regularity conditions provide useful guidelines to the construction of other posterior approximations in future research. 


\subsection{Ergodicity conditions and main result}
\label{sec:ergodic}
In this section, we follow the notation in \cite{adaptive_roberts2007} to formalize ADA.
Suppose the parameter space $\mathcal{X}$ is equipped with $\sigma$-algebra $\mathcal{B(X)}$.  
Without loss of generality, we define all the proposal densities, target densities and its approximations with respect to the Lebesgue measure here.

To simplify the notation, we use $\pi(\cdot) \equiv \pi_{\rm post}(\,\cdot\,|\,\tilde{\bf d})$ denote the exact posterior density.
%
%
We parametrize (potentially state-dependent) adaptive approximate posteriors and adaptive proposals using adaptation indices ${\boldsymbol\nu} \in \mathcal{V}$ and $\boldsymbol\xi \in \mathcal{E}$, respectively.
Here $\mathcal{V}$ and $\mathcal{E}$ denote the space of adaptation indices. 
We use the adaptation indices to replace the iteration number $n$ in Algorithm~\ref{al:adamh}, since the former provide unique notations for functions.
This way, we denote the adaptive approximate posterior as $\pi^*_{{\boldsymbol\xi},{\bf x}}(\cdot)$ and the adaptive proposal as $q_{\boldsymbol\nu}({\bf x},\cdot)$ in this section. 
For example, in Approximation~\ref{appr:aempost}, the adaptation index is $\boldsymbol\xi = (\bar{\boldsymbol\mu}_{B,n}, \bar{\boldsymbol\Sigma}_{B,n})$ that define the AEM, while in the GCAM algorithm~\ref{alg:gcam}, the adaptation indices are $\boldsymbol\nu = \{\sigma_{i}, \boldsymbol\Sigma_{n, \mathcal{I}_j}\}_{i = 1}^{L}$. 
%


%
%
%
%

%
%
We first introduce the transition kernels involved in ADA.

For each $\boldsymbol\nu\in\mathcal{V}$, we can define the transition kernel of a single level MH algorithm with target density $\pi(\cdot)$ as
\begin{equation*}
K_{\boldsymbol\nu}({\bf x, y}) = q_{\boldsymbol\nu}({\bf x, y}) \, \alpha_{\boldsymbol\nu}({\bf x, y}) + \Big[ 1 - \int_\mathcal{X}  q_{\boldsymbol\nu}( {\bf x, z} )\,\alpha_{\boldsymbol\nu}( {\bf x, z} ) d {\bf z} \Big] \delta_{\bf x}({\bf y}),
\end{equation*} 
where
\begin{equation*}
\label{eq:adaptive_1st}
\alpha_{\boldsymbol\nu}({\bf x, y}) = \min\left\{ 1, \frac{\pi({\bf y})\, q_{\boldsymbol\nu}\left({\bf y} ,{\bf x}\right)} {\pi({\bf x})\,q_{\boldsymbol\nu}\left({\bf x} ,{\bf y}\right)} \right\}.
\end{equation*} 
This way, $\{K_{\boldsymbol\nu}\}_{\boldsymbol\nu\in\mathcal{V}}$ defines a family of Markov chain transition kernels (associated with MH) on $\mathcal{X}$ with the target $\pi(\cdot)$.

For each pair of $\boldsymbol\nu\in\mathcal{V}$ and $\boldsymbol\xi\in\mathcal{E}$, we have the first step and second step acceptance probabilities
\[
\alpha_{\boldsymbol\nu,\boldsymbol\xi}({\bf x, y}) = \min\left\{ 1, \frac{\pi^*_{{\boldsymbol\xi},{\bf x}}({\bf y})\, q_{\boldsymbol\nu}\left({\bf y} ,{\bf x}\right)}
                     {\pi^*_{{\boldsymbol\xi},{\bf x}}({\bf x})\,q_{\boldsymbol\nu}\left({\bf x} ,{\bf y}\right)} \right\}, \textrm{\;and \;}
\beta_{\boldsymbol\nu,\boldsymbol\xi}({\bf x, y}) = \min\left\{1, \frac{\pi({\bf y}) \, q_{\boldsymbol\nu,\boldsymbol\xi}^*\left({\bf y} ,{\bf x}\right)}{\pi({\bf x}) \, q_{\boldsymbol\nu,\boldsymbol\xi}^*\left({\bf x} ,{\bf y}\right)}   
      \right\},
\] 
respectively, where the effective proposal in the second step has the density 
\begin{equation*}
   q_{\boldsymbol\nu,\boldsymbol\xi}^*\left({\bf x} ,{\bf y}\right) = 
     q_{\boldsymbol\nu}\left({\bf x} ,{\bf y}\right) \, \alpha_{\boldsymbol\nu,\boldsymbol\xi}({\bf x, y})  
     + \Big[ 1-\int_\mathcal{X}  q_{\boldsymbol\nu}\left({\bf x} ,{\bf z}\right)\, \alpha_{\boldsymbol\nu,\boldsymbol\xi}({\bf x, z}) \,\dd{\bf z}\Big] \delta_{\bf x}({\bf y}).
\end{equation*} 
We can define the transition kernel of ADA algorithm with target $\pi(\cdot)$ as
\begin{equation*}
K_{\boldsymbol\nu,\boldsymbol\xi}({\bf x, y}) = q_{\boldsymbol\nu}({\bf x, y}) \, \alpha_{\boldsymbol\nu,\boldsymbol\xi}({\bf x, y})\, \beta_{\boldsymbol\nu,\boldsymbol\xi}({\bf x, y}) + \Big[ 1 - \int_\mathcal{X}  q_{\boldsymbol\nu}( {\bf x, z} )\,\alpha_{\boldsymbol\nu,\boldsymbol\xi}( {\bf x, z} )\,\beta_{\boldsymbol\nu,\boldsymbol\xi}( {\bf x, z} ) d {\bf z} \Big] \delta_{\bf x}({\bf y}).
\end{equation*} 
This way, $\{K_{{\boldsymbol\nu},{\boldsymbol\xi}}\}_{{\boldsymbol\nu}\in\mathcal{V}, {\boldsymbol\xi}\in\mathcal{E}}$ defines a family of Markov chain transition kernels (associated with ADA) on $\mathcal{X}$ with target $\pi(\cdot)$.

In ADA, the adaptation indices ${\boldsymbol\nu}$ and ${\boldsymbol\xi}$ are respectively updated by a $\mathcal{V}$-valued random variable $\boldsymbol\Gamma_n$ and a $\mathcal{E}$-valued random variable $\boldsymbol\Xi_n$ at each step.
In contrast, DA only employs fixed parameters ${\boldsymbol\nu}$ and ${\boldsymbol\xi}$. 
Although for each pair of $\boldsymbol\nu\in\mathcal{V}$ and $\boldsymbol\xi\in\mathcal{E}$, the resulting DA scheme can be ergodic, there is no guarantee that the ADA scheme will be ergodic if the adaptation on ${\boldsymbol\nu}$ and ${\boldsymbol\xi}$ are not carefully constructed. 

Utilizing ergodic theory of standard adaptive MCMC that only adapts on the proposal, we aim to establish regularity conditions for ADA to be ergodic.
The key is to analyze the behaviour of the effective proposal $ q_{\boldsymbol\nu,\boldsymbol\xi}^*\left({\bf x} ,{\bf y}\right)$  in Definition 2 -- which involves both the proposal and the approximate posterior -- and the associated transition kernel $K_{\boldsymbol\nu,,\boldsymbol\xi}({\bf x, y})$ during adaptation.
By considering only the proposal adaptation (indexed by $\boldsymbol\nu\in\mathcal{V}$), Roberts and Rosenthal \cite{adaptive_roberts2007} provide conditions for constructing ergodic adaptive MCMC algorithms. 
We first restate Theorem 1 of Roberts and Rosenthal \cite{adaptive_roberts2007} with a small extension required for ADA, including both adaptation indices $\boldsymbol\nu\in\mathcal{V}$ and $\boldsymbol\xi\in\mathcal{E}$. 
\begin{theorem}
\label{theo1}
Suppose we have a target density $\pi(\cdot)$  defined on a parameter space $\mathcal{X}$. 
Suppose an ADA algorithm with $\mathcal{V}$-valued proposal adaptation index and $\mathcal{E}$-valued approximation adaptation index is ergodic for $\pi(\cdot)$ for given $\boldsymbol\nu\in\mathcal{V}$ and $\boldsymbol\xi\in\mathcal{E}$. 
Under the following conditions, ADA is ergodic:

\begin{enumerate}
\item (Simultaneous uniform ergodicity.) For all $\epsilon > 0$, there exist $n = n(\epsilon) \in \mathbb{N}$ such that 
\[
\| K_{\boldsymbol\nu,\boldsymbol\xi}^n({\bf x, \cdot}) - \pi(\cdot) \|_{TV} \leq \epsilon,
\]
for any ${\bf x} \in\mathcal{X}$, $\boldsymbol\nu\in\mathcal{V}$ and $\boldsymbol\xi\in\mathcal{E}$. Here $K^n_{\boldsymbol\nu,\boldsymbol\xi}({\bf x, y})$ is the $n$-step transition kernel defined as
\[
K^n_{\boldsymbol\nu,\boldsymbol\xi}({\bf x, y}) = \int_\mathcal{X}  K^{n-1}_{\boldsymbol\nu,\boldsymbol\xi}({\bf x, z}) K_{\boldsymbol\nu,\boldsymbol\xi}({\bf z, y}) d{\bf z}.
\]
\item (Diminishing adaptation.) In two consecutive iterations $n$ and $n+1$, the transition kernels satisfy
\[
\lim_{n \rightarrow \infty} \sup_{\bf x \in \mathcal{X}}\| K_{\boldsymbol\nu_{n+1},\boldsymbol\xi_{n+1}}({\bf x, \cdot}) - K_{\boldsymbol\nu_{n},\boldsymbol\xi_{n}}({\bf x, \cdot})\|_{TV} = 0.
\]
\end{enumerate}
Here $\|\lambda_1(\cdot) - \lambda_2(\cdot)\|_{TV} = \sup_{{\bf A}\in\mathcal{B}(\mathcal{X})}\|\lambda_1({\bf A})-\lambda_2({\bf A})\|$ is the total variational distance. \footnote{We denote $\lambda({\bf A}) = \int_{\bf A} \lambda({\bf x}) d {\bf x}$ for ${\bf A}\in\mathcal{B}(\mathcal{X})$.}
\end{theorem}
\begin{proof}
The proof directly follows from the Theorem 1 of Roberts and Rosenthal. \cite{adaptive_roberts2007}.
\end{proof}

As sufficient conditions for an MCMC with adaptive proposals satisfying simultaneous uniform ergodicity and diminishing adaptation are well understood in the literature, here we focus on establishing sufficient conditions on the adaptation of posterior approximations to make the resulting ADA ergodic. 
We prove the following two theorems that separately show ADA can satisfy simultaneous uniform ergodicity and diminishing adaptation by imposing mild regularity conditions. 



\begin{theorem}
\label{theo2}
Suppose an ADA algorithm with the target density $\pi(\cdot)$ has a family of first step proposal densities $\{ q_{\boldsymbol\nu}({\bf x, \cdot})\}_{{\boldsymbol\nu}\in\mathcal{V}}$, a family of approximate target densities $\{ \pi^*_{\boldsymbol\xi, x}({\cdot})\}_{{\boldsymbol\nu}\in\mathcal{V}, {\boldsymbol\xi}\in\mathcal{E}}$,  and a family of transition kernels $\{ K_{\boldsymbol\nu,\boldsymbol\xi}({\bf x, \cdot})\}_{{\boldsymbol\nu}\in\mathcal{V}, {\boldsymbol\xi}\in\mathcal{E}}$.
The transition kernels satisfy simultaneous uniform ergodicity given the following sufficient conditions. 
\begin{enumerate}
  \item The spaces $\mathcal{X}$, $\mathcal{V}$, and $\mathcal{E}$ are compact. \label{cond:a}
  \item The target density is Lipschitz continuous in ${\bf x}$. \label{cond:f}
  \item For any $\boldsymbol\nu \in \mathcal{V}$, if one applies the proposal $ q_{\boldsymbol\nu}({\bf x, \cdot})$ in a single-level MH with target $\pi(\cdot)$, then each transition kernel $K_{\boldsymbol\nu}$ (as defined in Definition 1) is ergodic for $\pi(\cdot)$. \label{cond:b}
  \item For any ${\boldsymbol\nu} \in\mathcal{V}$, the proposal $q_{\boldsymbol\nu}({\bf x},\cdot)$ is uniformly bounded, and for each fixed ${\bf y} \in \mathcal{X}$, the mapping $({\bf x},{\boldsymbol\nu}) \mapsto q_{\boldsymbol\nu}({\bf x} ,{\bf y})$ is Lipschitz continuous in ${\bf x}$ and $\boldsymbol\nu$. \label{cond:c}
  \item The mapping  $({\bf x},{\bf y},\boldsymbol\xi) \mapsto \log \pi^*_{\boldsymbol\xi,\bf x}({\bf y})$ is Lipschitz continuous in ${\bf x}$, ${\bf y}$, and $\boldsymbol\xi$. \label{cond:g}
\end{enumerate}
\end{theorem}
\begin{proof}
We extend the result of Corollary 5 of Roberts and Rosenthal \cite{adaptive_roberts2007} to prove this theorem. 
First note that, since $\mathcal{X}$, $\mathcal{Y}$, and $\mathcal{E}$ are compact, all product spaces are compact in the product topology.
Since for a given pair of indices $(\boldsymbol\nu,\boldsymbol\xi) \in\mathcal{V}\times\mathcal{E}$, the ADA becomes a standard DA, employing Theorem 1 of DA \cite{ChristenFox2005}, conditions~\eqref{cond:a}, \eqref{cond:b}, and \eqref{cond:g} imply that, an Markov chain induced by the transition kernel $K_{{\boldsymbol\nu},{\boldsymbol\xi}}$ is ergodic for $\pi(\cdot)$ for any pair of $(\boldsymbol\nu,\boldsymbol\xi) \in\mathcal{V}\times\mathcal{E}$. 

Recall that the effective proposal in the second step of ADA has density
\begin{equation*}
\label{eq:eff_proposal}
   q_{\boldsymbol\nu,\boldsymbol\xi}^*\left({\bf x} ,{\bf y}\right) = 
     q_{\boldsymbol\nu}\left({\bf x} ,{\bf y}\right) \,\alpha_{\boldsymbol\nu,\boldsymbol\xi}({\bf x, y}) 
     + r_{\boldsymbol\nu,\boldsymbol\xi}({\bf x}) \delta_{\bf x}({\bf y}).
\end{equation*} 
where $r_{\boldsymbol\nu,\boldsymbol\xi}({\bf x})$ is the probability of remaining at $\bf x$ after the first step, which is defined as
\[
r_{\boldsymbol\nu,\boldsymbol\xi}({\bf x}) = 1 - \int_\mathcal{X}  q_{\boldsymbol\nu}\left({\bf x} ,{\bf y}\right) \, \alpha_{\boldsymbol\nu,\boldsymbol\xi}({\bf x, z}) \,\dd{\bf z}.
\]
For each fixed ${\bf y} \in \mathcal{X}$, it follows from conditions \eqref{cond:c} and (\ref{cond:g}) that the map $({\bf x},{\boldsymbol\nu}, {\boldsymbol\xi})\mapsto q_{\boldsymbol\nu}({\bf x} ,{\bf y})\alpha_{{\boldsymbol\nu},{\boldsymbol\xi}}({\bf x}, {\bf y})$ is Lipschitz continuous, as $q_{\boldsymbol\nu}\left({\bf x} ,{\bf y}\right)\alpha_{\boldsymbol\nu,\boldsymbol\xi}({\bf x, y})$ takes the form 
\[
q_{\boldsymbol\nu}\left({\bf x} ,{\bf y}\right)\,\alpha_{\boldsymbol\nu,\boldsymbol\xi}({\bf x, y}) 
= q_{\boldsymbol\nu}\left({\bf x} ,{\bf y}\right)\,\min\left\{ 1, \frac{\pi^*_{{\boldsymbol\xi},{\bf x}}({\bf y})\, q_{\boldsymbol\nu}\left({\bf y} ,{\bf x}\right)}
                     {\pi^*_{{\boldsymbol\xi},{\bf x}}({\bf x})\,q_{\boldsymbol\nu}\left({\bf x} ,{\bf y}\right)} \right\}
= \min\left\{ q_{\boldsymbol\nu}\left({\bf x} ,{\bf y}\right)\,, \frac{\pi^*_{{\boldsymbol\xi},{\bf x}}({\bf y})}{\pi^*_{{\boldsymbol\xi},{\bf x}}({\bf x})} \, q_{\boldsymbol\nu}\left({\bf y} ,{\bf x}\right) \right\}.
\]
We also have the map $({\bf x},{\boldsymbol\nu}, {\boldsymbol\xi})\mapsto r_{{\boldsymbol\nu},{\boldsymbol\xi}}({\bf x})$ is Lipschitz continuous by the bounded convergence theorem.
%

Thus, for each fixed ${\bf y} \neq \bf x$, the second step acceptance probability of ADA
\[
\beta_{\boldsymbol\nu,\boldsymbol\xi}({\bf x, y}) = \min\left\{1, \frac{\pi({\bf y}) \, q_{\boldsymbol\nu,\boldsymbol\xi}^*\left({\bf y} ,{\bf x}\right)}{\pi({\bf x}) \, q_{\boldsymbol\nu,\boldsymbol\xi}^*\left({\bf x} ,{\bf y}\right)}   
      \right\} 
= \min\left\{1, \frac{\pi({\bf y}) \, q_{\boldsymbol\nu}\left({\bf y} ,{\bf x}\right)\,\alpha_{\boldsymbol\nu,\boldsymbol\xi}({\bf y, x}) }{\pi({\bf x}) \, q_{\boldsymbol\nu}\left({\bf x} ,{\bf y}\right)\,\alpha_{\boldsymbol\nu,\boldsymbol\xi}({\bf x, y}) }   
      \right\} ,
\]
is jointly continuous in ${\bf x}$, $\boldsymbol\nu$ and $\boldsymbol\xi$.
We also have the probability of remaining at $\bf x$ after both steps~\ref{step:adamh2} and \ref{step:adamh3} of ADA, which has the form 
\begin{equation*}
\rho_{{\boldsymbol\nu},{\boldsymbol\xi}}({\bf x}) = 1 - \int_\mathcal{X}  q_{\boldsymbol\nu}({\bf x} ,{\bf z})\,\alpha_{{\boldsymbol\nu},{\boldsymbol\xi}}({\bf x}, {\bf z})\, \beta_{{\boldsymbol\nu},{\boldsymbol\xi}}({\bf x}, {\bf z}) \, \dd{\bf z},
\end{equation*}
which is jointly continuous in ${\bf x}$, $\boldsymbol\nu$ and $\boldsymbol\xi$ by the bounded convergence theorem.

Denoting $k_{\boldsymbol\nu,\boldsymbol\xi}({\bf x, y}) = q_{\boldsymbol\nu}({\bf x, y}) \, \alpha_{\boldsymbol\nu,\boldsymbol\xi}({\bf x, y})\, \beta_{\boldsymbol\nu,\boldsymbol\xi}({\bf x, y}) $, we can decompose the transition kernel of ADA as
\[
K_{\boldsymbol\nu,\boldsymbol\xi}({\bf x, y}) = k_{\boldsymbol\nu,\boldsymbol\xi}({\bf x, y}) +  \rho_{{\boldsymbol\nu},{\boldsymbol\xi}}({\bf x}) \delta_{\bf x}({\bf y}).
\]
Note that the transition function $k_{\boldsymbol\nu,\boldsymbol\xi}({\bf x, y})$ is jointly continuous in ${\bf x}$, $\boldsymbol\nu$ and $\boldsymbol\xi$ for each fixed $\bf y \in \mathcal{X}$ following the above derivations. 
Since the Dirac delta $\delta_{\bf x}({\bf y})$ is a point mass, $\delta_{\bf x}({\bf y})$ and the Lebesgue measure are orthogonal measures.
This way, iterating the transition kernel, we have the $n$-step transition kernel
\[
K^n_{\boldsymbol\nu,\boldsymbol\xi}({\bf x, y}) = k^n_{\boldsymbol\nu,\boldsymbol\xi}({\bf x, y}) +  \rho_{{\boldsymbol\nu},{\boldsymbol\xi}}({\bf x})^n \delta_{\bf x}({\bf y}),
\]
in which the $n$-step transition function $k^n_{\boldsymbol\nu,\boldsymbol\xi}({\bf x, y})$ is also joint continuous in ${\bf x}$, $\boldsymbol\nu$ and $\boldsymbol\xi$.
Furthermore, we have
\[
\| K_{\boldsymbol\nu,\boldsymbol\xi}^n({\bf x, \cdot}) - \pi(\cdot) \|_{TV} = \rho_{{\boldsymbol\nu},{\boldsymbol\xi}}({\bf x})^n + \frac12 \int_\mathcal{X}  \big| k^n_{\boldsymbol\nu,\boldsymbol\xi}({\bf x, y}) - \pi({\bf y}) \big| d {\bf y},
\]
following the property of total variation distance that $\| \lambda_1(\cdot) - \lambda_2(\cdot) \|_{TV} = \frac12 \int_\mathcal{X}  |\lambda_1({\bf x}) - \lambda_2({\bf x})| d{\bf x}$. 
This quantity is again joint continuous in ${\bf x}$, $\boldsymbol\nu$ and $\boldsymbol\xi$ by the bounded convergence theorem.
In addition, for each fixed pair of $\boldsymbol\nu$ and $\boldsymbol\xi$, $\lim_{n \rightarrow \infty}\| K_{\boldsymbol\nu,\boldsymbol\xi}^n({\bf x, \cdot}) - \pi(\cdot) \|_{TV}$ uniformly converges to zero in ${\bf x}$, $\boldsymbol\nu$ and $\boldsymbol\xi$ by ergodicity and compactness.
Therefore, the simultaneous uniform ergodicity condition holds given uniformly convergence and continuity.
\end{proof}

\begin{theorem}
\label{theo3}
Suppose an ADA algorithm with the target density $\pi(\cdot)$ has a family of transition kernels $\{ K_{\boldsymbol\nu,\boldsymbol\xi}({\bf x, \cdot})\}_{{\boldsymbol\nu}\in\mathcal{V}, {\boldsymbol\xi}\in\mathcal{E}}$.
Suppose further Conditions 1--5 of Theorem \ref{theo2} hold. 
The transition kernel satisfies the diminishing adaptation condition given the following conditions:
\begin{enumerate}
  \item The proposal satisfies diminishing adaptation, that is, \label{cond:d}
  \[
  \lim_{n\rightarrow\infty}\sup_{\bf x}\|q_{\boldsymbol\nu_{n+1}}({\bf x},\cdot)-q_{\boldsymbol\nu_{n}}({\bf x},\cdot)\|_{TV} = 0 \textrm{\quad in\;probability}, 
  \]   
  \item The approximation adaptation index satisfies diminishing adaptation, that is, $\lim_{n\rightarrow\infty} \| \boldsymbol\Xi_{n+1}-\boldsymbol\Xi_{n} \| = 0$ in probability.  \label{cond:h}
\end{enumerate}
\end{theorem}

\newcommand{\transada}[1]{p_{\boldsymbol\nu_{#1},\boldsymbol\xi_{#1}}({\bf x, y})\, \beta_{\boldsymbol\nu_{#1},\boldsymbol\xi_{#1}}({\bf x, y})}

\begin{proof}
The first half of our proof uses the result of Lemma 4.21 in~\cite{LRR2013}. 
For simplicity, we define the transition probability of ADA as $p_{\boldsymbol\nu,\boldsymbol\xi}({\bf x, y}) = q_{\boldsymbol\nu}({\bf x, y}) \, \alpha_{\boldsymbol\nu,\boldsymbol\xi}({\bf x, y})$, which is the probability propose $\bf y$ from $\bf x$ and accept $\bf y$ in the first step of ADA.
This way, the transition kernel of the ADA has the form
\begin{equation*}
K_{\boldsymbol\nu,\boldsymbol\xi}({\bf x, y}) = p_{\boldsymbol\nu,\boldsymbol\xi}({\bf x, y})\, \beta_{\boldsymbol\nu,\boldsymbol\xi}({\bf x, y}) + \rho_{{\boldsymbol\nu},{\boldsymbol\xi}}({\bf x}) \delta_{\bf x}({\bf y}),
\textrm{\; where \;}
\rho_{{\boldsymbol\nu},{\boldsymbol\xi}}({\bf x}) = 1 - \int_\mathcal{X}  p_{{\boldsymbol\nu},{\boldsymbol\xi}}({\bf x}, {\bf z})\, \beta_{{\boldsymbol\nu},{\boldsymbol\xi}}({\bf x}, {\bf z}) \, \dd{\bf z}.
\end{equation*}
Then, for any $\bf x \in \mathcal{X}$ and $\bf A \in \mathcal{B}(\mathcal{X})$, transition kernels in two consecutive iterations $n$ and $n+1$ satisfy
\begin{eqnarray}
\big| K_{\boldsymbol\nu_{n+1},\boldsymbol\xi_{n+1}}({\bf x, A}) - K_{\boldsymbol\nu_{n},\boldsymbol\xi_{n}}({\bf x, A}) \big| \hspace{-10em} && \nonumber \\
&=&  \left| \int_{\bf A} \left[ \transada{n+1} - \transada{n} \right] d {\bf y} \right. \nonumber \\
&& \left. + \mathbbm{1}_{\bf x \in A} \left[ \rho_{{\boldsymbol\nu}_{n+1},{\boldsymbol\xi}_{n+1}}({\bf x}) - \rho_{{\boldsymbol\nu}_{n},{\boldsymbol\xi}_{n}}({\bf x}) \right] \right| \nonumber \\
&\leq& \int_{\bf A}  \left| \transada{n+1} - \transada{n} \right| d {\bf y}  \nonumber \\
%
%
&& + \mathbbm{1}_{\bf x \in A}  \int_\mathcal{X} \left| \transada{n+1} - \transada{n} \right| d{\bf y} ,
\label{eq:theo31}
\end{eqnarray}
where $\mathbbm{1}_{\bf x \in A}$ is the indicator function. For $\bf x \neq y$, we have
\[
p_{\boldsymbol\nu,\boldsymbol\xi}({\bf x, y}) \,\beta_{\boldsymbol\nu,\boldsymbol\xi}({\bf x, y}) = \min\left\{p_{\boldsymbol\nu,\boldsymbol\xi}({\bf x, y}), \frac{\pi({\bf y}) }{\pi({\bf x}) }\, p_{\boldsymbol\nu,\boldsymbol\xi}({\bf y, x}) \right\} .
\]
Given the identity $|\min\{a, b\} - \min\{c, d\} | \leq |a - c| + |b - d|$, we have
\begin{equation}
\left| \transada{n+1} - \transada{n} \right| \leq \left|p_{\boldsymbol\nu_{n+1},\boldsymbol\xi_{n+1}}({\bf x, y}) - p_{\boldsymbol\nu_{n},\boldsymbol\xi_{n}}({\bf x, y})\right| \left(1 + \frac{\pi({\bf y}) }{\pi({\bf x}) }\right).
\label{eq:theo32}
\end{equation}
The compactness of parameter space $\mathcal{X}$ and continuity of the target $\pi(\cdot)$ imply that 
\begin{equation}
\left(1 + \frac{\pi({\bf y}) }{\pi({\bf x}) }\right) < C_1, \; \forall {\bf x, y} \in \mathcal{X},
\label{eq:theo33}
\end{equation}
for some constant $C_1 < \infty$.
Recall the property $\| \lambda_1(\cdot) - \lambda_2(\cdot) \|_{TV} = \frac12 \int_\mathcal{X} |\lambda_1({\bf y}) - \lambda_2({\bf y})| d{\bf y}$, Eqn. \ref{eq:theo31}--\ref{eq:theo33} imply that
\begin{equation*}
\big| K_{\boldsymbol\nu_{n+1},\boldsymbol\xi_{n+1}}({\bf x, A}) - K_{\boldsymbol\nu_{n},\boldsymbol\xi_{n}}({\bf x, A}) \big| \leq C_2\,\left\|p_{\boldsymbol\nu_{n+1},\boldsymbol\xi_{n+1}}({\bf x, \cdot}) - p_{\boldsymbol\nu_{n},\boldsymbol\xi_{n}}({\bf x, \cdot})\right\|_{TV},
\end{equation*}
for some constant $C_2 < \infty$. Thus, we have
\begin{equation*}
\left\|K_{\boldsymbol\nu_{n+1},\boldsymbol\xi_{n+1}}({\bf x, \cdot}) - K_{\boldsymbol\nu_{n},\boldsymbol\xi_{n}}({\bf x, \cdot})\right\|_{TV} \leq C_2\,\left\|p_{\boldsymbol\nu_{n+1},\boldsymbol\xi_{n+1}}({\bf x, \cdot}) - p_{\boldsymbol\nu_{n},\boldsymbol\xi_{n}}({\bf x, \cdot})\right\|_{TV}.
\end{equation*}

Since we can bound the total variation distance between transition kernels by the total variation distance between transition densities in the form of $p_{\boldsymbol\nu,\boldsymbol\xi}({\bf x, y}) = q_{\boldsymbol\nu}({\bf x, y}) \, \alpha_{\boldsymbol\nu,\boldsymbol\xi}({\bf x, y})$, diminishing adaptation of the overall transition kernel follows from diminishing adaptation of $q_{\boldsymbol\nu}({\bf x, y}) \, \alpha_{\boldsymbol\nu,\boldsymbol\xi}({\bf x, y})$.

Given the triangle inequality 
\begin{equation*}
\left\|p_{\boldsymbol\nu_{n+1},\boldsymbol\xi_{n+1}}({\bf x, \cdot}) - p_{\boldsymbol\nu_{n},\boldsymbol\xi_{n}}({\bf x, \cdot})\right\|_{TV} 
\leq 
\left\|p_{\boldsymbol\nu_{n+1},\boldsymbol\xi_{n+1}}({\bf x, \cdot}) - p_{\boldsymbol\nu_{n},\boldsymbol\xi_{n+1}}({\bf x, \cdot})\right\|_{TV} + \left\|p_{\boldsymbol\nu_{n},\boldsymbol\xi_{n+1}}({\bf x, \cdot}) - p_{\boldsymbol\nu_{n},\boldsymbol\xi_{n}}({\bf x, \cdot})\right\|_{TV},
\end{equation*}
we can treat adaptation indices ${\boldsymbol\nu}$ and ${\boldsymbol\xi}$ in separate steps. 
Using a similar argument in the first part of this proof, we can show that for any fixed approximation adaptation index $\boldsymbol\xi$, diminishing adaptation of $p_{{\boldsymbol\nu},{\boldsymbol\xi}}({\bf x},{\bf y})$ with respect to proposal adaptation in $\boldsymbol\nu$ follows directly from the above condition \ref{cond:d}
In the rest of this proof, we establish the diminishing adaptation of $p_{{\boldsymbol\nu},{\boldsymbol\xi}}({\bf x},{\bf y})$ with respect to approximation adaptation in $\boldsymbol\xi$.
For any fixed proposal adaptation index $\boldsymbol\nu$, we have the following inequality 
\begin{eqnarray*}
 \left\|p_{{\boldsymbol\nu},{\boldsymbol\xi}_{n+1}}({\bf x},\cdot) - p_{{\boldsymbol\nu},{\boldsymbol\xi}_n}({\bf x},\cdot)\right\|_{TV} & = & \frac12 \int_\mathcal{X} q_{\boldsymbol\nu}({\bf x, y}) \, \left| \alpha_{\boldsymbol\nu,\boldsymbol\xi_{n+1}}({\bf x, y}) - \alpha_{\boldsymbol\nu,\boldsymbol\xi_{n}}({\bf x, y}) \right| d{\bf y} \nonumber \\
 & = & \frac12 \int_\mathcal{X} \left| \min\left\{ q_{\boldsymbol\nu}({\bf x, y}), \frac{\pi^*_{{\boldsymbol\xi}_{n+1},{\bf x}}({\bf y})}{\pi^*_{{\boldsymbol\xi}_{n+1},{\bf x}}({\bf x})} q_{\boldsymbol\nu}\left({\bf y} ,{\bf x}\right) \right\} \right. \nonumber\\
 && \left.- \min\left\{ q_{\boldsymbol\nu}({\bf x, y}), \frac{\pi^*_{{\boldsymbol\xi}_{n},{\bf x}}({\bf y})}{\pi^*_{{\boldsymbol\xi}_{n},{\bf x}}({\bf x})} q_{\boldsymbol\nu}\left({\bf y} ,{\bf x}\right) \right\} \right| d{\bf y} \nonumber\\
 & \leq & \frac12 \int_\mathcal{X} \left| \frac{\pi^*_{{\boldsymbol\xi}_{n+1},{\bf x}}({\bf y})}{\pi^*_{{\boldsymbol\xi}_{n+1},{\bf x}}({\bf x})} - \frac{\pi^*_{{\boldsymbol\xi}_{n},{\bf x}}({\bf y})}{\pi^*_{{\boldsymbol\xi}_{n},{\bf x}}({\bf x})} \right| q_{\boldsymbol\nu}\left({\bf y} ,{\bf x}\right) d {\bf y}.
\end{eqnarray*}
Given the compactness of $\mathcal{X} \times \mathcal{E}$ and Lipschitz continuity of $\log \pi^*_{\boldsymbol\xi,\bf x}({\bf y})$, it follows that the RHS uniformly converges to zero as $\|{\boldsymbol\xi}_{n+1}-{\boldsymbol\xi}_{n}\|\rightarrow 0$.
Thus, diminishing adaptation holds.
\end{proof}


Conditions required in Theorems \ref{theo2} and \ref{theo3} are not restrictive for many practical applications. 
Condition \eqref{cond:a} of Theorem \ref{theo2}, that parameter space and the adaptation spaces are compact, is often a consequence of physical bounds on the parameters and bounded model outputs. In practical computation, one could argue that this assumption always holds as computers are finite dimensional, though one does not want to explore the full range of numerical representations if the algorithm is to be efficient! 
Conditions \eqref{cond:b} and \eqref{cond:c} of Theorem \ref{theo2} and Condition \eqref{cond:d} of Theorem \ref{theo3} are conditions on on the proposal distribution, and depend on the choice of proposal and adaptation that is used. Thus, these conditions can be satisfied by making suitable choices. 
Conditions \ref{cond:f} and \ref{cond:g} of Theorem \ref{theo2}, of Lipschitz continuity, is satisfied by the forward model
and its reduced model in most inverse problems; Indeed, the more ill-posed is the inverse problem, the more well-posed is the forward model and the higher the order of continuity satisfied by the forward model. Finite-dimensional reduced models are Lipschitz continuous because stiffness matrices are not singular when the reduced model is well posed.

Then, we can use Theorems~\ref{theo2} and \ref{theo3} to establish ergodicity of ADA using the approximation schemes in Section~\ref{sec:approx}.

\begin{corollary}
\label{coro1}
Suppose that the forward map $F(\cdot)$ and its reduced model $F^*(\cdot)$ are continuous functions, then the ADA Algorithm~\ref{al:adamh} with either of the adaptive proposals in Section~\ref{sec:adapt}, and using any of Approximation~\ref{appr:coarse} to \ref{appr:aemsd} in Section~\ref{sec:approx} is ergodic for $\pi(\cdot)$. 
\end{corollary}
\begin{proof}
We treat Approximations~\ref{appr:aempost} and \ref{appr:aemsd}, since Approximations~\ref{appr:coarse}, \ref{appr:aempr}, and \ref{appr:local} are non-adaptive special cases.

Without loss of generality, we may assume that the parameter space $\mathcal{X}$ is compact, as we can always define bounds on the input parameter to the computer model. Since $F(\cdot)$ and $F^*(\cdot)$ are continuous, it follows that the model reduction error $B(\cdot)$ or $B_{{\bf x}}(\cdot)$ is compact. Thus, the space of possible $\boldsymbol\mu_{B,n}$  and $\boldsymbol\Sigma_{B,n}$ is compact, that is, $\mathcal{E}$ is compact. Compactness of $\mathcal{V}$ follows from the form of proposal adaptation in Algorithms~\ref{alg:am} and \ref{alg:gcam}, and existing results for such proposals, such as Corollary 6 of  \cite{adaptive_roberts2007}. This establishes Condition~\eqref{cond:a} of Theorem~\ref{theo2}. 

Since $F(\cdot)$ and $F^*(\cdot)$ are continuous, Condition~\eqref{cond:g} of Theorem~\ref{theo2}  follows from the form of Equations~\ref{eq:a_enhan_post} and \ref{eq:a_enhan_lp} when $\boldsymbol\Sigma_{B,n}+\boldsymbol\Sigma_{\bf e}$ is nonsingular. Since $\boldsymbol\Sigma_{\bf e}$ is positive definite and $\boldsymbol\Sigma_{B,n}$ is positive semi-definite, it follows that $\boldsymbol\Sigma_{B,n}+\boldsymbol\Sigma_{\bf e}$ is actually positive definite.

The AEM updating rules Eqns~\ref{eq:updatemu1}--\ref{eq:updatesigma1} and Eqn. \ref{eq:appr41} satisfy the diminishing adaptation condition (Condition~\eqref{cond:h} of Theorem~\ref{theo3}), because the empirical estimates change $O(1/n)$, in probability, at the $n$-th iteration. Similarly, the proposal diminishing adaptation condition (Condition~\eqref{cond:d} of Theorem~\ref{theo3}) follows from the form of adaptation in Algorithms~\ref{alg:am} and \ref{alg:gcam}.

The conditions in Theorems~\ref{theo2} and \ref{theo3} hold, and so the result follows. 
\end{proof}


\section{Applications in fitting geothermal reservoir models}
\label{sec:geo}

In this section we apply ADA to two calibration problems for geothermal reservoir models. First, a one dimensional radial symmetry model of the feedzone of a geothermal reservoir with synthetic data is presented. This example is small enough that extensive statistics can be computed to study relative efficiencies of algorithms. Then ADA is applied to sample a 3D model with measured data. We begin with a description of the governing equations  of the geothermal reservoir and its numerical simulator.

\subsection{Data simulation}
\label{sec:data_simu}
Consider a two phase geothermal reservoir (water and vapour) governed by the general mass balance and energy balance equations\cite{geo_book,geo_osullivan_1985}
\begin{equation} 
\frac{d}{dt} \int_{\Omega} M_\alpha \,dV = \int_{\partial \Omega} Q_\alpha \cdot {\bf n} \, d\Gamma + \int_{\Omega} q_\alpha \, dV ,
\quad \alpha\in\{\textrm{m},\textrm{e}\},
\label{eq:conserve}
\end{equation} 
where $\Omega$ is the control volume and $\partial \Omega$ is its boundary. The accumulation term $q_\alpha$ represents the mass ($q_\textrm{m}$) and heat ($q_\textrm{e}$) sources or sinks in $\Omega$, and $Q_\alpha$ denotes the mass ($Q_\textrm{m}$) or energy ($Q_\textrm{e}$) flux through $\partial \Omega$. The mass and energy within $\Omega$ are represented by $M_\textrm{m}$ and $M_\textrm{e}$.

A complex set of nonlinear partial differential equations including non-isothermal, multiphase Darcy's law are used to model $M_\alpha$ and $Q_\alpha$. 
%
%
Using a two-phase flow as the example, the mass and energy per unit volume can be modelled by
\begin{eqnarray}
\label{eq:mass}
  \displaystyle M_\textrm{m} & = & \displaystyle \phi \, \big[\rho_\textrm{l}
(1 - S_\textrm{v}) + \rho_\textrm{v} S_\textrm{v}\big] , \\
  \displaystyle M_\textrm{e} & = & \displaystyle (1 - \phi) \, \rho_\textrm{r}
c_\textrm{r} T + \phi \, \big[\rho_\textrm{l} u_\textrm{l}(1 - S_\textrm{v}) +
\rho_\textrm{v} u_\textrm{v} S_\textrm{v}\big] ,
\label{eq:energy}
\end{eqnarray}
where $\phi$ is the porosity of rock, $T$ is the temperature, and $S_{\rm v}$ represent the vapour saturation. 
We use the subscripts $l$, $v$ and $r$ represent the liquid phase, the vapour phase, and the rock, respectively. 
There are physical constants involved in the above equations: $c_{(\cdot)}$ is the specific heat, $\rho_{(\cdot)}$ denotes the density, and $u_{(\cdot)}$ is the specific internal energy.
%
%
The mass and energy fluxes are modelled by 
\begin{eqnarray}
\label{eq:qm}
\displaystyle Q_{\textrm{m}} & = & \displaystyle \sum_{\beta=\textrm{l,v}}
\displaystyle \frac{k k_{\textrm{r}\beta}}{\nu_{\beta}}(\triangledown p -
\rho_{\beta}\vec{g}), \\
  Q_\textrm{e} & = & \displaystyle \sum_{\beta=\textrm{l,v}} \displaystyle \frac{k
k_{\textrm{r}\beta}}{\nu_{\beta}}(\triangledown p -
\rho_{\beta}\vec{g})h_{\beta} - K\triangledown T ,
\label{eq:qe}
\end{eqnarray}
where $k$ is a diagonal second order permeability tensor in 3-dimensions. Physical constants $h_{(\cdot)}$ and $K$ denote specific enthalpy and the thermal conductivity in a saturated medium, respectively.
In the above equations, relative permeabilities $k_{\rm rl}$ and $k_{\rm rv}$---which are empirically derived functions---are introduced to account for the interference between liquid and vapour phases.
Here, we use the van Genuchten-Mualem model \cite{vangenuchten}:
\[
(k_\textrm{rl}, k_\textrm{rv}) = f_\textrm{vGM}(S_\textrm{v}; m,S_\textrm{rl}, S_\textrm{ls}),
\] 
which is a function of the saturation $S_\textrm{v}$. The van Genuchten-Mualem model also depends on explicitly provided parameters $m$, $S_\textrm{rl}$, and $S_\textrm{ls}$ to characterize the behaviour of relative permeabilities.

In the system of equations \ref{eq:conserve}--\ref{eq:qe}, the system states are the spatially distributed quantities $(p, T, S_{\rm v})$. We can make partial observation either on these quantities directly or some measurements that are typically  nonlinear function of these quantities. 
The system of equations contains unknown parameters including porosity, permeability, initial and boundary conditions, and parameters of the relative permeability model, are fit to data using sample-based inference.

We use the package TOUGH2 \cite{tough2} to numerically solve the system of equations \ref{eq:conserve}--\ref{eq:qe}, for a given set of parameters, using the finite volume method with first order accuracy in space.
TOUGH2 carries time integration using the implicit first order scheme, in which the implicit time integration is solved by the Newton--Krylov method. In TOUGH2, time stepping is automatically adjusted according to the number of Newton iterations in each time step. 
This way, numerical models based on grids with different resolutions can be naturally used to construct the forward model and its reduced model. We do not control the time step in the discretization. 

\subsection{Well discharge test analysis}
\label{sec:1d}

\subsubsection{Parameter and data}
\begin{figure}[ht]
\centering
\includegraphics[width=0.8\textwidth]{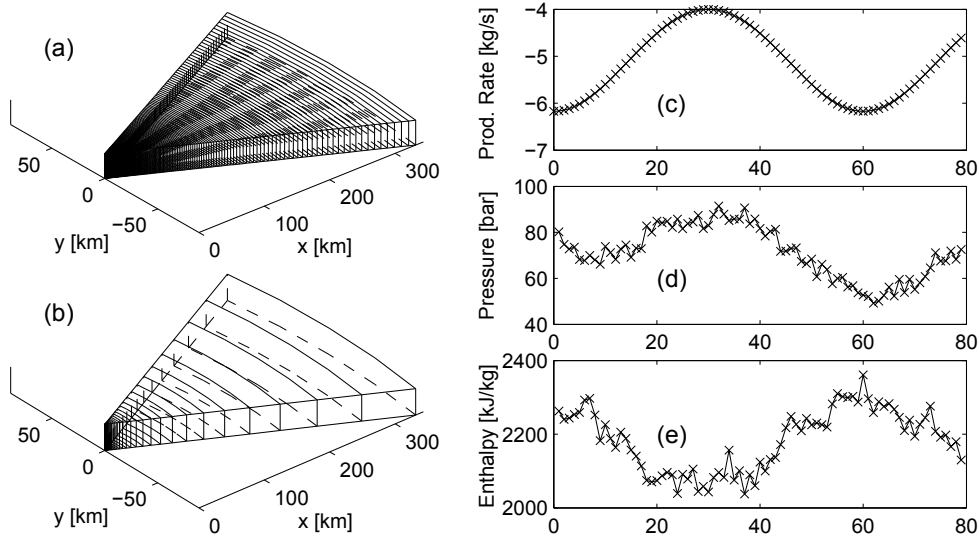}
\caption{Finite volume grids used for well discharge test analysis and data sets used for well discharge test. (a): the forward model (640 blocks), (b): the reduced model (40 blocks), (c): the production rate (kg/second), (d): the pressure (bar), and (e): the flowing enthalpy (kJ/kg).}
\label{figure:wellall} 
\end{figure}

Well discharge analysis is usually used to interpret the near-well properties of the reservoir from pressure and enthalpy data measured during a short period of field production. 
Based on the typical assumption that all flows into the well come through a single layer feedzone, we build an one-dimensional radially-symmetric forward model $F(\cdot)$ with 640 blocks as shown in Fig.~\ref{figure:wellall} (a).
A high resolution grid is used immediately outside the well, with cell thickness increasing exponentially outside this region. The reduced model $F^*(\cdot)$ is built by coarsening the grid of the forward model to a coarse grid with 40 blocks, see Figure \ref{figure:wellall} (b). 
Based on 1,000 simulations with different set of parameters on a DELL T3400 workstation, we estimate the CPU time for evaluating the forward model is 2.60 seconds. The computing time of the reduced model is 0.15 seconds, which is $5.8\%$ of that for the forward model.

\begin{table}[ht]
\caption{Prior constraints for parameters of the well test model, and parameter values for generating synthetic data.}
\label{table:prior}
\centering
   \begin{tabular}{llllllll}
   \toprule
   & $\phi$[-] & $\log_{10}(k) [\textrm{m}^2]$ & $p_\textrm{0}$ [bar] & $S_\textrm{v0}$ [-] & $m$ [-] & $S_\textrm{rl}$ [-] & $S_\textrm{ls}$ [-]  \\
   \midrule
   Lower bound & 0 & -Inf & 0.5 & 0 & 0.7 & 0  &  100 \\
   \midrule
   Upper bound & 0.3 & Inf & 1 & 0.5 & 1 & 0.3  & 150 \\
   \midrule
   True value& 0.12 & -14.82 & 120 & 0.1 & 0.65 & 0.25 & 0.91 \\
   \bottomrule
   \end{tabular}
\end{table}

The parameters of interest are the porosity, (horizontal) permeability, the parameters in the van Genuchten-Mualem relative permeability model, as well as the initial vapour saturation ($S_\textrm{v0}$) and initial pressure ($p_\textrm{0}$) that are used to represent the initial thermodynamic state of the two-phase system. These make up the seven unknowns required for data simulation:
\begin{equation*}
  {\bf x} = \left\{ {\phi}, \log_{10}(k), p_\textrm{0}, S_\textrm{v0},  m, S_\textrm{rl}, S_\textrm{ls}\right\}.
\end{equation*}
Note that the permeability $k$ is represented on a base 10 logarithmic scale. These parameters are assumed to be independent and follow non-informative prior distributions with bounds given in Table~\ref{table:prior}. Table \ref{table:prior} also gives ``true'' values of model parameters used to simulate the synthetic data. The model is simulated over 80 days with production rates varying smoothly from about 4 kg/second to about 6 kg/second (see Figure \ref{figure:wellall} (c)).

For well test experiments, the observations include the pressure and flowing enthalpy measured at $N$ discrete time points $t_1, t_2, \ldots, t_N$ up to day 30, where the flowing enthalpy $h_f$ is a nonlinear function of pressure and vapour saturation, i.e., $h_f( p, S_{\rm v} )$. This way, we can denote the observed data as
\[
\tilde{\bf d} = \big[\tilde{p}(t_1), \tilde{p}(t_2), \ldots, \tilde{p}(t_N), \tilde{h}_f(t_1),\tilde{h}_f(t_2),\ldots, \tilde{h}_f(t_N) \big]^\top.
\]
For each realization of the model parameter, we can simulate the forward model $F(\bf x)$ to generate observable model outputs corresponding to the observed data.

We assume the measurement noise follows i.i.d. zero mean Gaussian distribution with standard deviations $\sigma_p = 3\;\textrm{bar}$ for pressure and $\sigma_h = 30\;\textrm{kJ/kg}$ for the flowing enthalpy. The noise corrupted pressure and flowing enthalpy data are plotted in Figure \ref{figure:wellall} (d) and (e), respectively. This yields the posterior distribution
\begin{eqnarray*}
\pi_{\rm post}({\bf x} \, | \, \tilde{\bf d}) & \propto & \exp\left[-\frac{1}{2}  \left( F({\bf x})-\tilde{\bf d} \right)^\top  {\bf \Sigma}_\textrm{\bf e}^{-1} \left( F({\bf x})- \tilde{\bf d}\right)\right] \, \chi({\bf x}) 
\label{eq:well_post}
\end{eqnarray*} 
where $\chi({\bf x} )$ is the prior distribution which we take to be the indicator function that implements the bounds in Table \ref{table:prior} and 
\[
{\bf \Sigma}_\textrm{\bf e} = \left[\begin{array}{cc} \sigma_\textrm{h}^2 {\bf I}_n & 0 \\ 0 & \sigma_\textrm{p}^2 {\bf I}_n \end{array}\right]
\]
is the covariance of measurement noise.


To check the assumptions of Theorems \ref{theo2} and \ref{theo3}, we note that the compactness of the parameter space is satisfied as we impose bounds on parameters. 
Lipschitz continuity of the forward model and reduced model, and hence the continuity of the posterior and its approximations based on reduced model, is satisfied as the system of equations \ref{eq:conserve}--\ref{eq:qe} are continuous. 
Compactness of the approximation adaptation space $\mathcal{E}$ directly follows from the compactness of the parameter space and the continuity of the forward model and reduced model.

\subsubsection{MCMC sampling}

To benchmark various approximate posterior distributions, we first run GCAM with one group of all 7 parameters and targeting the exact posterior distribution. Based on several short runs (not reported here), we find that an acceptance rate of $13\%$ gives optimal efficiency for this problem. This configuration of GCAM is run for $10^5$ iterations (after discarding $5\times10^4$ burn-in steps) giving an IACT of the log-likelihood function estimated as $84.4$. 
Then we then run ADA using Approximation~\ref{appr:coarse}, Approximation~\ref{appr:aempr} (with the AEM calculated \emph{a priori}), Approximation~\ref{appr:aempost} (with the AEM calculated adaptively over the posterior), and Approximation~\ref{appr:aemsd} with the AEM calculated adaptively over the posterior. All cases used GCAM for the proposal with the target acceptance rate of $13\%$.
The acceptance rate in step~\ref{step:adamh3} of ADA, $\bar{\beta}$, and the IACT of the likelihood function are shown in Table~\ref{table:performance}. 
We note that when Approximation~\ref{appr:coarse} is used in ADA, the method is the equivalent method in~\cite{Efendiev}, except that we use an adaptive proposal here.

\begin{table}[ht]
\caption{Performance summary of various approximations for the well test model. The abbreviations used are: $\bar{\beta}$ -- the second step acceptance rate, and IACT -- IACT of the log-likelihood function.}
\label{table:performance}
\centering
\begin{tabular}{llllll}
\toprule
          & Approx. 1 & Approx. 2 (prior)  & Approx. 3 (posterior) & Approx. 5 & Standard MH \\
\midrule
$\bar{\beta}$   & $0.12$ & $0.31$ & $0.77$   & $0.93$   & $1$     \\
IACT            & -      & -      & $208$    & $153$    & $169$ \\
\bottomrule
\end{tabular}
\end{table}

Approximation~\ref{appr:coarse} (approximation uses reduced model directly) only produces $\bar{\beta}=12 \%$. (This agrees closely with the equivalent method in~\cite{Efendiev}.) Approximation~\ref{appr:aempr} with AEM built over the prior, increases the acceptance rate in step~\ref{step:adamh3} of ADA to $\bar{\beta}=31 \%$.
However, both Approximation~\ref{appr:coarse} and the AEM constructed over the prior cannot produce a well mixed Markov chain, even after $2\times10^5$ iterations, so the IACT for the log-likelihood function could not be estimated, and is not reported for these cases.
Approximation~\ref{appr:aempost} with the AEM calculated adaptively over the posterior produces significantly better mixing, with an estimated $\bar{\beta}=77 \%$, and the IACT of the log-likelihood function is  $208$. 
Approximation \ref{appr:aemsd}, with AEM calculated adaptively over the posterior and with state-dependent reduced model \eqref{eq:rom_l}, improves the performance further, achieving $\bar{\beta}=93 \%$, and the IACT of the log-likelihood function is $153$. 
For Approximation \ref{appr:aempost} and \ref{appr:aemsd}, we ran both chains for $2\times10^6$ iterations to calculate further summary statistics, with the first $5\times10^5$ steps discarded as burn-in.

\begin{figure}[h!]
\centering
\includegraphics[width=0.8\textwidth]{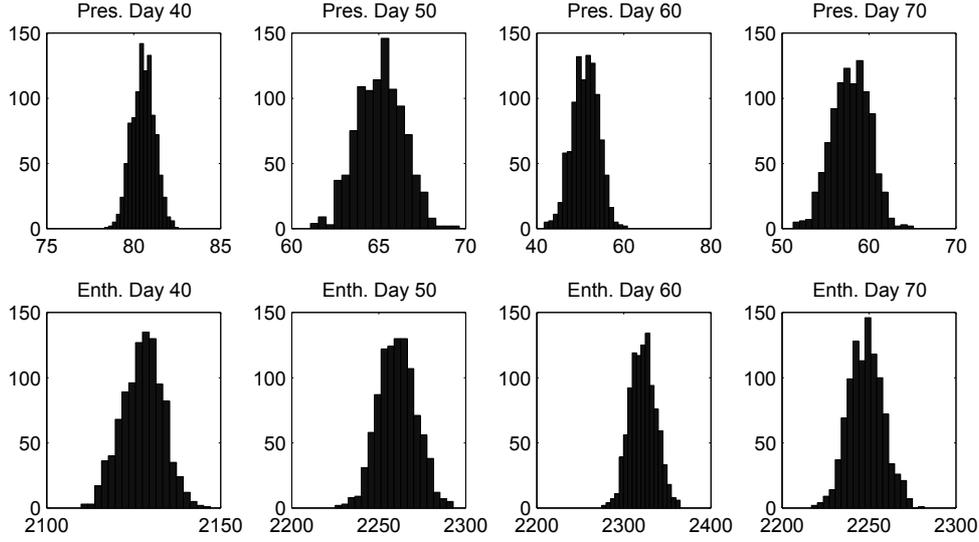}
\caption{Histograms of the predictive density, at day 40, 50, 60 and 70. Top row: Pressure, bottom row: flowing enthalpy.}
\label{figure:s_pred_hist}
\end{figure}

\begin{figure}[h!]
\centering
\includegraphics[width=0.45\textwidth]{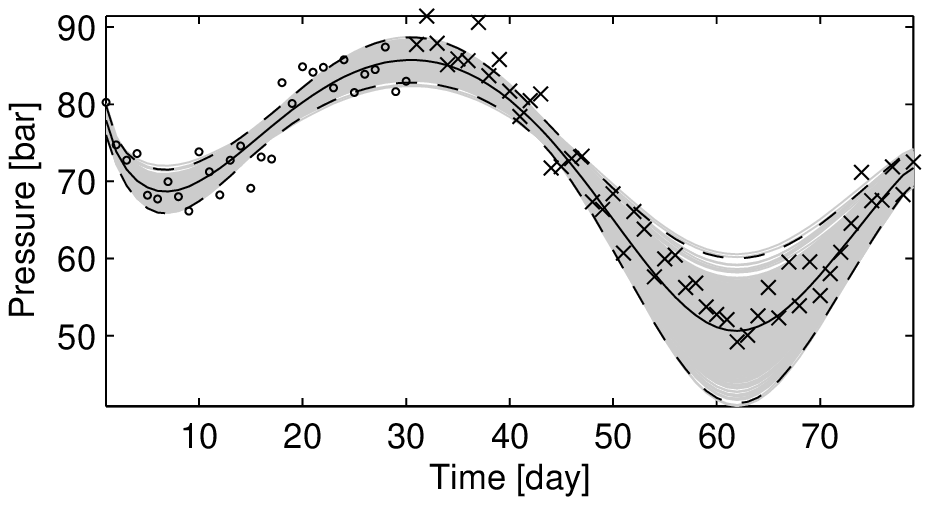}
\includegraphics[width=0.45\textwidth]{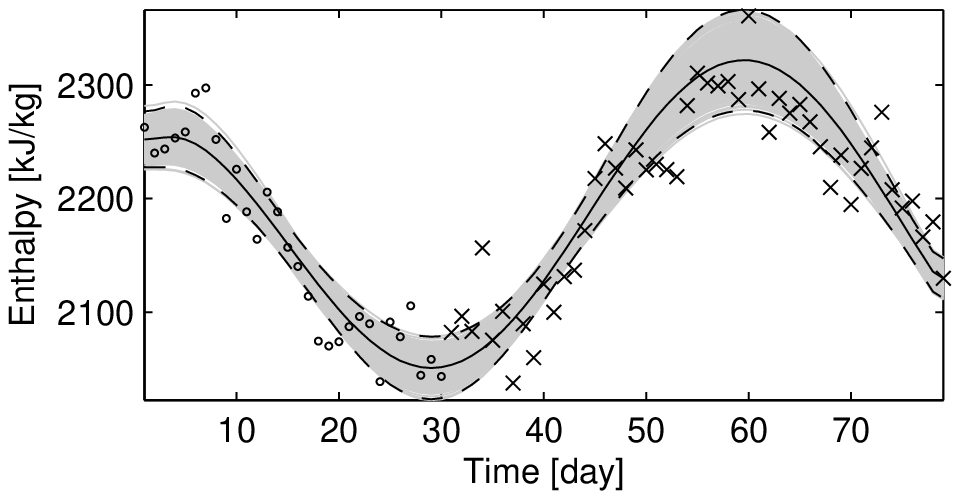}
\caption{Predictions for pressure and flowing enthalpy. Left column: pressure, Right column: flowing enthalpy. 
The circles and the crosses are the training data and validation data, respectively; the solid line and dashed lines are the mean prediction and $95 \%$ credible interval, respectively; and the shaded lines represent the predictions made by samples. 
}
\label{figure:s_pred}
\end{figure}

By using formula \eqref{eq:su_factor1}, we can estimate that the factor by which computational efficiency is improved for ADA with Approximation~\ref{appr:aempost} (posterior) and \ref{appr:aemsd} is about $4.3$ and $5.9$, respectively. 
In contrast, the use of Approximation~\ref{appr:aempr} (prior) only improves computational efficiency marginally, while the use of the na\"ive Approximation~\ref{appr:coarse} appears to actually reduce computational efficiency.
We also notice that the IACTs of the log-likelihood function suggest that the ADA with Approximation~\ref{appr:aemsd} is more statistically efficient than the standard MH, which cannot be the case as discussed in Section \ref{sec:cme}. This effect could be caused by finite-sampling error in the IACT estimate. 
However, this result suggests that the decrease of statistical efficiency may be negligible in this particular case.

The histograms of the model predictions computed on several different time points are given in Figure \ref{figure:s_pred_hist}, with pressure in the top row and enthalpy in the bottom row. We can observe that the predictions at these observation times follow uni-modal distributions. The model predictions and the $95 \%$ credible intervals over an 80-day period are shown in Figure \ref{figure:s_pred}. For both predictions, the means follow the observed data reasonably well.

\begin{figure}[ht]
\centering
\includegraphics[width=0.8\textwidth]{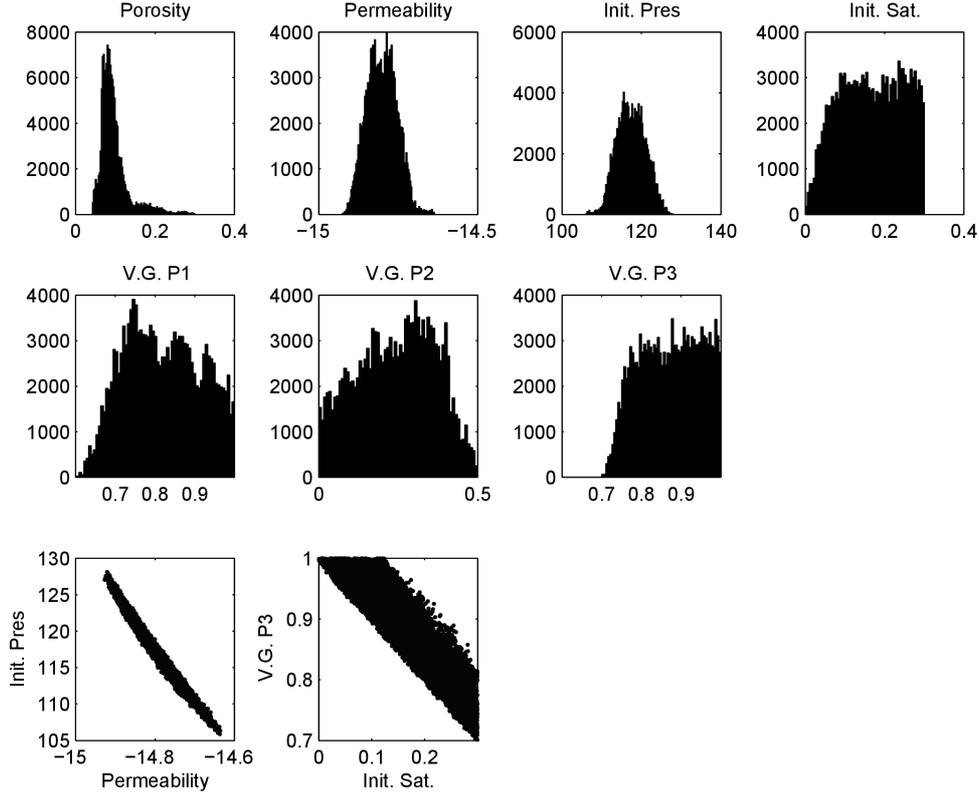}
\caption{Histograms of the marginal distributions and scatter plots between parameters, for the synthetic data set.}
\label{figure:s_hist}
\end{figure}

The histograms of the marginal distributions (first two rows of Figure \ref{figure:s_hist}) of the parameter $\bf x$ show skewness in porosity and two of the parameters of the van Genuchten-Mualem relative permeability model ($m$ and $S_\textrm{rl}$). The scatter plots between parameters show strong negative correlations between the permeability (on base $10$ logarithmic scale) and the initial pressure; see the left plot of last row of Figure \ref{figure:s_hist}. There is also a strong negative correlation between the initial saturation and one of the hyperparameters of the van Genuchten-Mualem ($S_\textrm{ls}$); see the right plot of last row of Figure \ref{figure:s_hist}.

\subsection{Natural state modelling}
\label{sec:3d}

\subsubsection{Parameter and data}

We now present an application of ADA to a 3D  geothermal reservoir model with measured field data. 
We aim to infer the  permeability structure and mass input at the bottom of the reservoir from temperature data measured from wells.
We also wish to predict the size and shape of the hot plume of the reservoir, which is only sparsely measured in wells.
The forward model covers a volume of $12.0$ km by $14.4$ km extending down to $3050$ meters below sea level. 
Relatively large blocks were used near the outside of the model and then were progressively refined near the wells to achieve a well-by-well allocation to the blocks. 
The 3D structure of the forward model $F(\cdot)$ has $26,005$ blocks, and is shown in Figure~\ref{figure:mokai} (a), where the blue lines in the middle of the grid show wells drilled into the reservoir. 
To speed up the computation a reduced model $F^*(\cdot)$ based on a coarse grid with $3,335$ blocks is constructed by combining adjacent blocks in the $x$, $y$ and $z$ directions of the forward model; see Figure \ref{figure:mokai} (b). 
Each simulation of the forward model takes about $30$ to $50$ minutes CPU time on a DELL T3400 workstation, and the computing time for the reduced model is about $1$ to $1.5$ minutes (roughly $3\%$ of the forward model). The computing time for these models is sensitive to the input parameters.

\begin{figure}[ht]
\centering
\includegraphics[width=0.8\textwidth]{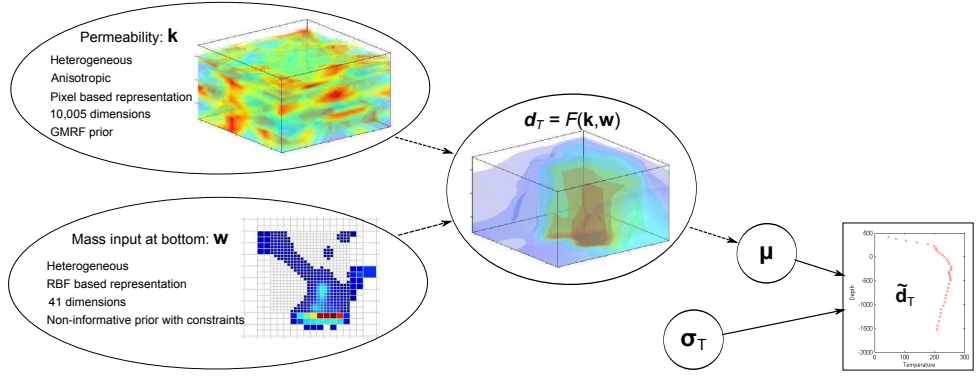}
\caption{3D modelling}
\label{figure:ns_3d}
\end{figure}

\begin{figure}[ht]
\centering
\includegraphics[width=0.8\textwidth]{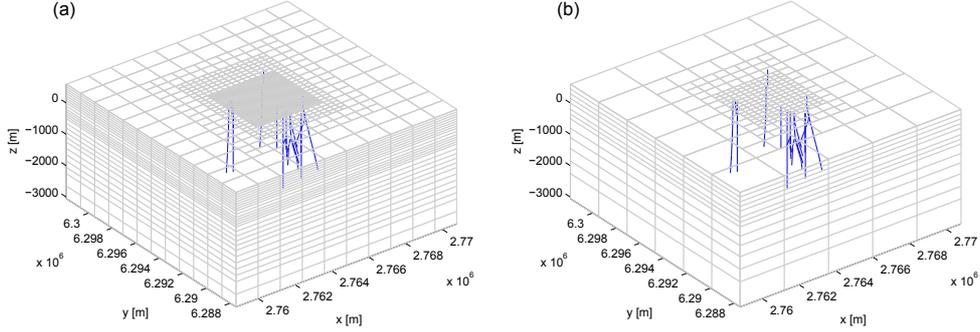}
\caption{The fine grid (left) and the coarse grid (right) used for natural state modelling.}
\label{figure:mokai}
\end{figure}

Let $s$ denote the spatial coordinate and $\Omega \subset \mathbb{R}^3$ denote the domain of the reservoir.
In the natural state modelling, we are interested in estimating spatially distributed and heterogeneous permeabilities $k(s)$ for $s \in \Omega$ and the mass input $q_{\rm m}(r)$ from the bottom boundary $\partial \Omega_{b} \subset \mathbb{R}^2$, where $r \in \partial \Omega_{b}$. 

We model the permeability tensor by scalar-valued permeability functions in the form of
\[
k(s) = \left[ \begin{array}{lll} k_{1}(s) && \\ & k_{2}(s) & \\ & & k_{3}(s) \end{array}\right],
\]
and discretize permeability functions $k_{1}(s)$, $k_{2}(s) $, and $k_{3}(s)$ on the coarse grid using piecewise linear representation. 
Given the centres of the finite volume cells in the coarse model, $s_1, s_2, \ldots, s_{3335}$, we have the discretized permeabilities
\[
{\bf k} = \big[k_{1}(s_1), k_{1}(s_2), \ldots, k_{1}(s_{3335}), k_{2}(s_1), k_{2}(s_2), \ldots, k_{2}(s_{3335}),  k_{3}(s_1), k_{3}(s_2), \ldots, k_{3}(s_{3335})\big]^\top,
\]
which is about $10,005$ dimensional. 
The spatial correlation of permeabilities in the base 10 logarithmic scale are modelled by a Gaussian Markov random field prior \cite{gmrf_rue}. We further impose that the permeability are bounded between $10^{-2}$ and $10^3$ millidarcy. 
This leads to the prior distribution of ${\bf k} $ in the form of
\[
\pi_{\rm prior}({\bf k}) \propto \exp\left( - \frac12 \log_{10}({\bf k})^\top \, {\bf Q} \, \log_{10}({\bf k}) \right)\,\chi({\bf k}),
\]
where $\bf Q$ is a symmetric positive definite sparse matrix constructed from the Gaussian Markov random field.
The mass input from the bottom boundary is modeled by the linear combination of radial basis functions with the squared-exponential kernel function
\[
q_{\rm m}(r) = \sum_{i=1}^{41} w_i \exp\left[ - \lambda \, \| r - r_i\|^2 \right],
\] 
centred at pre-specified control points ${\bf r}_i, i = 1,\ldots,41$. 
This way, the unknowns in this parametrization are the $41$ dimensional weighting variable 
\[
{\bf w} = (w_1,\ldots,w_{41})^\top,
\]
associated with the control points. 
Since this weighting vector $\bf w$ controls the distribution of the mass input, the constraints ${\bf w} > 0$ and $\sum {\bf w} = constant$ are imposed to ensure that the mass input is positive and has a fixed total amount.
Overall, we have unknown parameters $\bf x = [k, w]^\top$ with prior
\begin{eqnarray}
\label{eq:mokaipr}
\pi_{\rm prior}({\bf x})  & = \pi_{\rm prior}({\bf k}) \, \pi_{\rm prior}({\bf w}) \propto & \exp\left( - \frac12 \log_{10}({\bf k})^\top \, {\bf Q} \, \log_{10}({\bf k}) \right) \,\chi({\bf k})\,\chi({\bf w}),  \\
& \rm subject \ to & \rm \sum {\bf w} = constant \ and \ {\bf w} > 0, \nonumber
\end{eqnarray}
where $\chi({\bf k})$ and $\chi({\bf w})$ impose bounds on parameters. 
We refer interested readers to \cite{wrr_cfo_2011} for a further details of the prior modelling of this problem.

Steady state temperature measured at discrete locations along the well bores (as shown by blues in Figure \ref{figure:mokai}) are used for estimating permeability and mass input at the depth. 
Given measurement locations $s_1, s_2, \ldots, s_{N}$, we have the observed data
\[
\tilde{\bf d} = \big[\tilde{T}(s_1),\tilde{T}(s_2),\ldots,\tilde{T}(s_N) \big]^\top,
\]
as shown by the crosses in Figure  \ref{figure:temp}. 
Empirical estimation of the noise vector \cite{wrr_cfo_2011} suggests an i.i.d. Gaussian distribution with standard deviation $\sigma_\textrm{T} = 7.5^{\circ}$C to be used in the likelihood function. 
Simulating the forward model $F$ with realizations of the parameter $\bf x$ produces observable model outputs that are temperatures at measurement locations. 
This yields the posterior distribution
\begin{eqnarray}
\pi_{\rm post}({\bf x}|\tilde{\bf d}) & \propto & \exp\left[-\frac{1}{2{\sigma_\textrm{T}}^2}  \left\|F({\bf x}) - \tilde{\bf d} \right\|_2^2 \right] \, \pi_{\rm prior}({\bf x}) 
\label{eq:3d_post}
\end{eqnarray} 
where $\pi_{\rm prior}({\bf x})$ is defined in Eqn. \ref{eq:mokaipr}. Similar to the well test case, compactness assumptions and continuity assumptions (of the posterior and its approximations) in Theorems \ref{theo2} and \ref{theo3} are satisfied. 

\subsubsection{MCMC Sampling}

In the first step of ADA, the new permeability state is proposed using GCAM, and the weights controlling the mass input are proposed separately by an adaptive reversible jump move to meet the specific constraints. These proposals produce a first step acceptance rate of $\bar{\alpha}=0.1$ in ADA; see \cite{wrr_cfo_2011} for details.
Since the forward model is computationally very demanding, we first simulate the chain with the reduced model only, for about 200 sweeps of updates to get through the initial burn-in period. Then, we start ADA with Approximation~\ref{appr:aemsd} to sample the exact posterior distribution. 
 
We are able to sample the posterior distribution for about $11,200$ iterations in $40$ days, and ADA achieves about $\bar{\alpha}=74\%$ acceptance rate in the second accept-reject step. 
After discarding the first  $2,800$ iterations as burn-in steps, the estimated IACT of the log-likelihood function is about $5.6$. This is only a rough estimate because the chain has not been running long enough, however all the samples are consistent with measured data, as shown by the model outputs of the realizations (Figure \ref{figure:temp}).
Using the computational cost of the fine model and the IACT estimated from MCMC simulations over the posterior distribution defined using the coarse model, we can estimate that the adaptive delayed acceptance achieves a speed-up factor of $7.7$.

The mean and standard deviation of the temperature profiles are estimated as posterior sample averages. We compare these estimates with the measured data in Figure \ref{figure:temp}. 
The solid black lines are the estimated mean temperatures, dashed black lines denote the $95\%$ percent credible interval, and measured data are shown as red crosses. The green and gray lines represent outputs of the forward model and various reduced model realizations, respectively.
Figure \ref{figure:temp} shows that the forward model and the reduced model produce significantly different temperature profiles, and the forward model is hotter than the reduced model in average. 
This suggests that the model outputs are defined on some low dimensional manifold of the data space, and the forward model and reduced model produce outputs that are located on substantially different manifolds. 
The mean model reduction error at each of the measurement positions span a range of $[-33.64, 54.98]$, and hence the noise level of these model reduction errors are more significant than the zero mean normal distribution with standard deviation $\sigma_e = 7.5\,^{\circ}\mathrm{C}$. 
Therefore, the stochastic modeling of the model reduction error in Section \ref{sec:approx} is essential for efficient and  accurate inference when using the reduced model.

\begin{figure}[ht]
\centering
\includegraphics[width=0.8\textwidth]{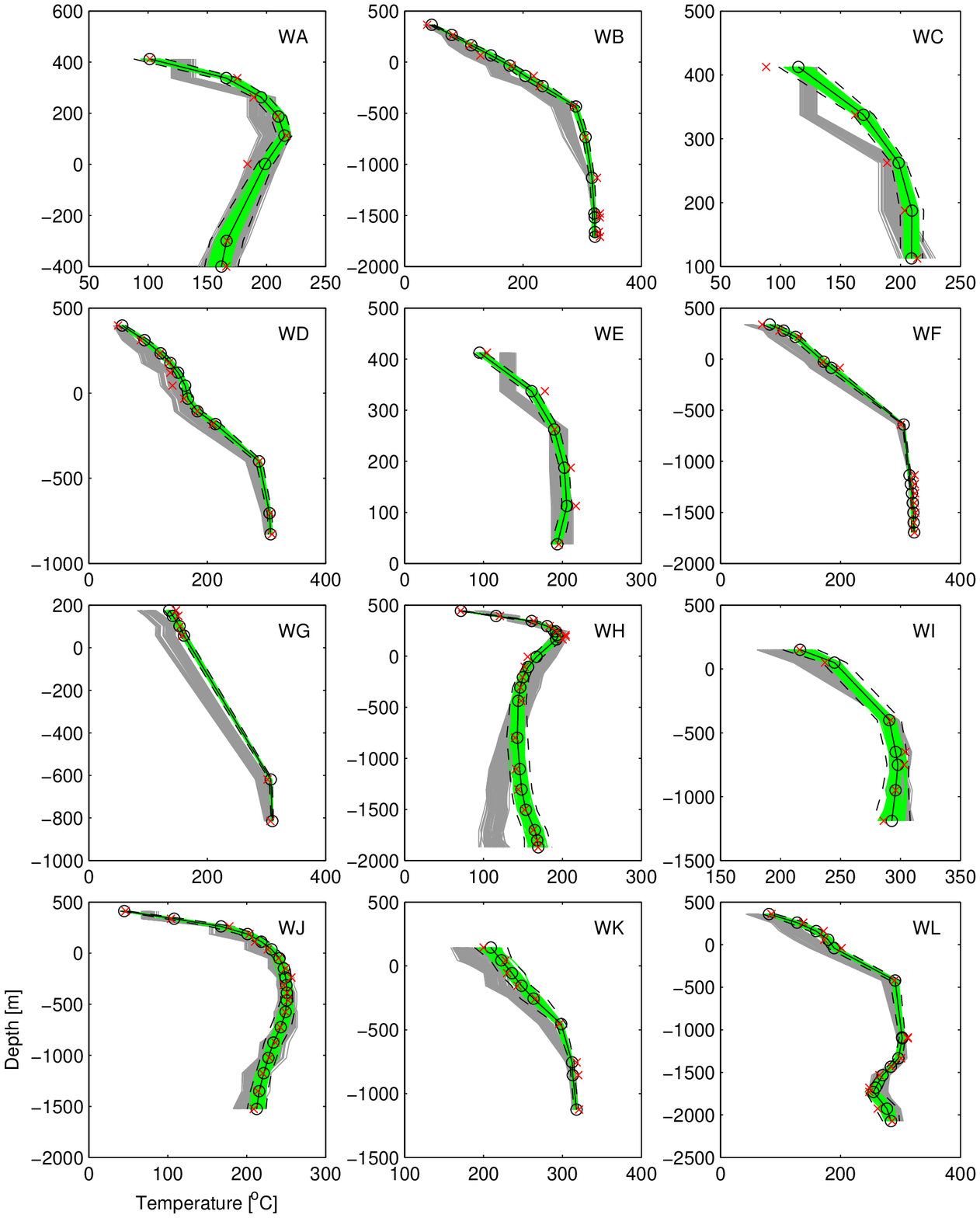}
\caption{Comparison of estimated temperatures and measured data. The solid black lines are the estimated mean temperatures, dashed black lines are the $95\%$ percent credible interval, and measured data are shown as red crosses. The green and gray lines represent outputs of the forward model and the reduced model various realizations, respectively.}
\label{figure:temp}
\end{figure}

\begin{figure}[ht]
\centering
\includegraphics[width=0.9\textwidth,height=0.4\textheight]{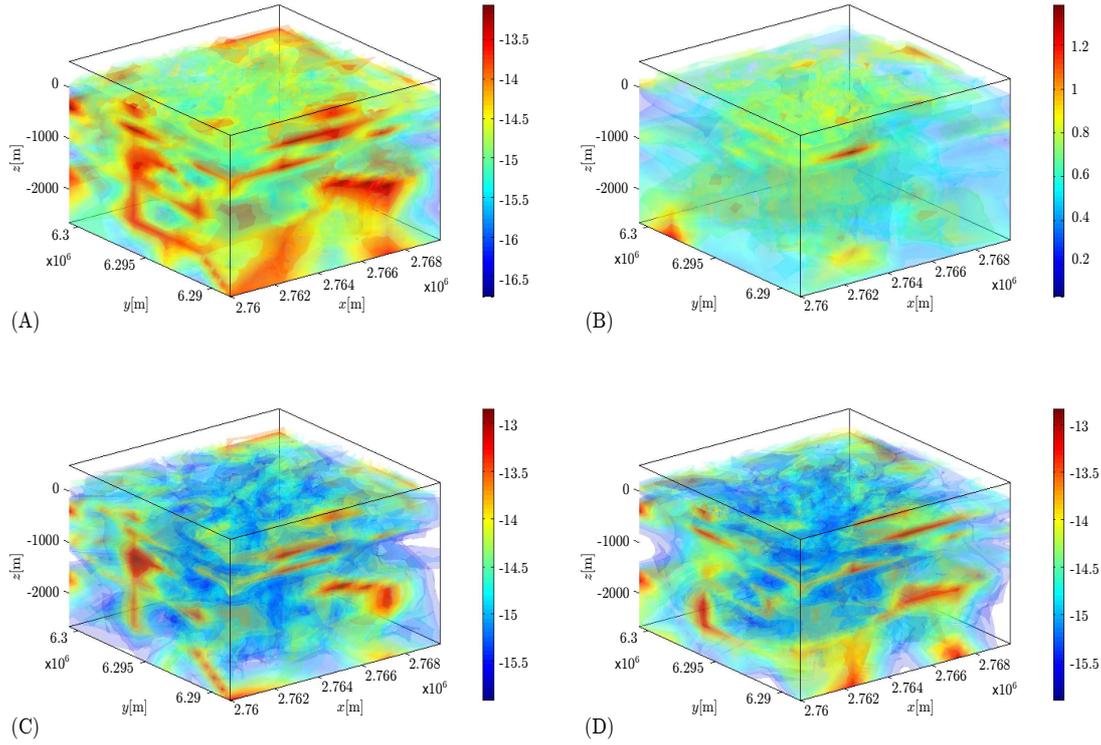}
\caption{Permeability distribution in the vertical direction on base 10 logarithmic scale. (a): the sample mean, (b): the sample standard deviation, (c): one realization from the Markov chain, and (d): another realization from the Markov chain.}
\label{figure:gk}
\end{figure}

\begin{figure}[ht]
\centering
\includegraphics[width=\textwidth,height=0.4\textheight]{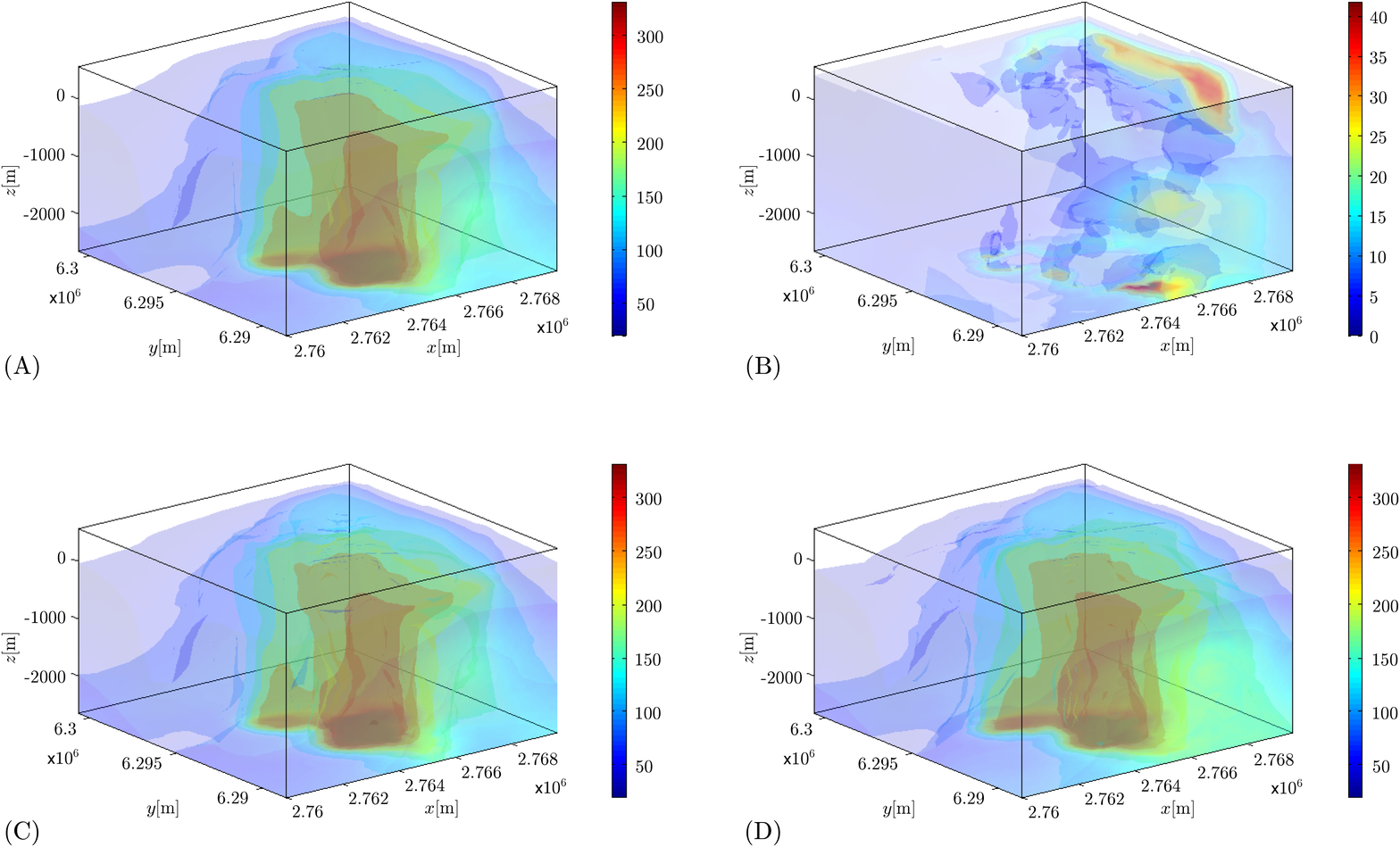}
\caption{Distributions of model temperatures, unit is $^{\circ}C$. (a): the sample mean, (b): the sample  standard deviations, (c): one realization from the Markov chain, and (d): another realization from the Markov chain.}
\label{figure:temp_d}
\end{figure}

The samples of the permeability distributions (we only report the permeabilities in the vertical direction for brevity, see Figure \ref{figure:gk}) are highly variable, even though each of these samples produces reasonable match to the measured data. 
In comparison the predicted temperature distribution (Figure \ref{figure:temp_d}) has very small uncertainty intervals. It reveals that even though the parameters are not well identified, the predicted temperature field is relatively well determined.

\section{Discussion}
\label{sec:dis}
We considered sample-based uncertainty quantification for inverse problems within the Bayesian formulation. Our primary contribution has been to advance computational methods that utilize a cheap approximation to the forward map, or reduced model, to improve computational efficiency of sample-based inference, that is the computational rate of reducing Monte Carlo errors in sample-based estimates. 
We presented state-dependent and stochastic corrections to reduced models, and gave a novel algorithm for adapting to the AEM, along with regularity conditions that guarantee ergodicity for the correct posterior distribution. We validated this algorithm and the improved approximations in two examples of parameter estimation for multi-phase flow in geothermal reservoirs, including a large-scale natural-state system using field data.

A straightforward method is to use a reduced model, such as a coarse PDE solve, in place of the correct forward map, and employ the DA algorithm to ensure the correct target distribution. We presented this as Approximation~\ref{appr:coarse} and computed results using this approximation. In our test cases, 
we found that the acceptance rate, $\bar{\beta}$ in the second accept-reject step in DA (or ADA) was never more than $0.2$. This indicates that the reduction in statistical efficiency, that necessarily occurs when using an approximation, may effectively cancel the reduction in compute time per iteration, resulting in negligible improvement in computational efficiency for the multi-phase flow problems we considered.

We improved the approximation to the target distribution, at no significant increase in computational cost, by introducing a state-dependent correction to the coarse PDE solve, and also a stochastic model for the model reduction error. The former was made possible by using the DA algorithm~\cite{ChristenFox2005}, while posterior estimation of the AEM was made possible by using adaptive MCMC techniques~\cite{adaptive_roberts2007}. We validated the sequence of approximations in a stylized seven-dimensional inverse problem in multi-phase flow, that allowed extensive evaluation of diagnostics.

Of particular importance was adaptive construction of the AEM over the posterior distribution; construction of the AEM over the prior, as in the original formulation of~\cite{stats_inverse}, improves statistical efficiency noticeably by increasing $\bar{\beta}$ from $0.12$ to $0.31$, but the construction over the posterior made a much greater improvement  to $\bar{\beta}=0.77$, while also not needing the expensive pre-calculation required for prior construction of the AEM. Including both the state-dependent correction and posterior AEM further gave further improvement with to $\bar{\beta}=0.93$, indicating that statistical efficiency is only slightly reduced in comparison to standard MH, so reduction in computational cost per iteration is transferred directly to improvement in computational efficiency.

Finally, we reported results from a large-scale inverse problem in calibration and prediction using a comprehensive numerical model of an actual geothermal field, using measured field data. For that example, the novel ADA algorithm with Approximation~\ref{appr:aemsd} gave a increase in computational efficiency by a factor of about $7.7$. Given the very large-scale nature of the computation required in this example, the computational savings by using the novel ADA algorithm are significant. We expect that the computational approach we have developed here will also be of significant value for uncertainty quantification in other large-scale inverse problems in science and engineering.

The approximations used in this paper are based on grid coarsening, due to the blackbox nature of our forward model. 
The power of our adaptive delayed acceptance framework is not limited by this form of reduced models. 
The regularity conditions on the approximate posteriors established in Theorems \ref{theo2} and \ref{theo3} offer further insights in designing new goal-oriented reduced models for solving inverse problems. 
For example, methods such as projection-based model reduction\cite{mor_BGW_2015, MOR_DEIM_reservoir, mor_LWG_2015, mor_CMW_2015, mor_APL_2016}, which often replies on expensive offline training phase by repeatedly evaluating froward models over prior samples, can also be adaptively constructed using the ADA framework.
This can potentially lead to significant speed-up and accuracy improvement for quantifying uncertainties of inverse problems.

\section*{Acknowledgement}
T. Cui acknowledges financial support from the Australian Research Council, under grant number LP170100985.

\bibliographystyle{acm}
\bibliography{CGref}

\end{document}